\documentclass[aps,prx,twocolumn,preprintnumbers,superscriptaddress]{revtex4} 

\usepackage{soul}
\usepackage[T1]{fontenc}
\usepackage{amstext,amsmath,amssymb,amsthm}
\usepackage[normalem]{ulem}
 \usepackage{soul}
\usepackage{color}
\usepackage{float} 
\usepackage{times}
\usepackage{algorithm}
\floatname{algorithm}{Protocol}
\usepackage{ftnxtra}
\usepackage{enumitem}
\usepackage{graphicx}
\usepackage[colorlinks]{hyperref}

\usepackage{soul}

\allowdisplaybreaks[2]

\newcounter{thm}

\theoremstyle{plain}

\newtheorem{theorem}[thm]{Theorem}
\newtheorem{corollary}{Corollary}
\newtheorem{lemma}{Lemma}
\newtheorem{prop}{Proposition}
\theoremstyle{definition}
\newtheorem{definition} {Definition}

\newcommand{\beq}{\begin{equation}}
\newcommand{\eeq}{\end{equation}}
\newcommand{\ket} [1] {\vert #1 \rangle}
\newcommand{\bra} [1] {\langle #1 \vert}

\newcommand{\Tr}{\text{Tr}}

\newcommand{\ba}{\begin{align}}
\newcommand{\ea}{\end{align}}
\newcommand{\bea}{\begin{eqnarray}}
\newcommand{\eea}{\end{eqnarray}}

\makeatletter

\@ifundefined{textcolor}{}
{
 \definecolor{BLACK}{gray}{0}
 \definecolor{WHITE}{gray}{1}
 \definecolor{RED}{rgb}{1,0,0}
 \definecolor{GREEN}{rgb}{0,.6,0}
 \definecolor{BLUE}{rgb}{0,0,1}
 \definecolor{CYAN}{cmyk}{1,0,0,0}
 \definecolor{MAGENTA}{cmyk}{0,1,0,0}
 \definecolor{YELLOW}{cmyk}{0,0,1,0}
 }

\makeatother

\usepackage{bm}
\usepackage{mathtools}

\def\id{I}

\def\1{\mat{\id}}

\def\mat#1{\vec{#1}}

\renewcommand{\vec}[1]{\bm{\mathrm{#1}}}

\renewcommand{\sout}[1]{}

\begin{document}
\title{Quantum noise protects quantum classifiers against adversaries}

\date{\today}

\author{Yuxuan Du}
\affiliation{UBTECH Sydney AI Centre, School of Computer Science, Faculty of Engineering, University of Sydney, Australia}
\author{Min-Hsiu Hsieh}
\affiliation{Centre for Quantum Software and Information, Faculty of Engineering and Information Technology, University of Technology Sydney, Australia}
\author{Tongliang Liu}
\affiliation{UBTECH Sydney AI Centre, School of Computer Science, Faculty of Engineering, University of Sydney, Australia}
\author{Dacheng Tao}
\affiliation{UBTECH Sydney AI Centre, School of Computer Science, Faculty of Engineering, University of Sydney, Australia}
\author{Nana Liu}
\email{Nana.Liu@quantumlah.org}
\affiliation{John Hopcroft Center for Computer Science, Shanghai Jiao Tong University, China}

\begin{abstract}

Noise in quantum information processing is often viewed as a disruptive and difficult-to-avoid feature, especially in near-term quantum technologies. However, noise has often played beneficial roles, from enhancing weak signals in stochastic resonance to protecting the privacy of data in differential privacy. It is then natural to ask, can we harness the power of quantum noise that is beneficial to quantum computing? An important current direction for quantum computing is its application to machine learning, such as  classification problems. One outstanding problem in machine learning for classification is its sensitivity to adversarial examples. These are small, undetectable perturbations from the original data where the perturbed data is completely misclassified in otherwise extremely accurate classifiers. They can also be considered as `worst-case' perturbations by unknown noise sources. We show that by taking advantage of depolarisation noise in quantum circuits for classification, a robustness bound against adversaries can be derived where the robustness improves with increasing noise. This robustness property is intimately connected with an important security concept called differential privacy which can be extended to quantum differential privacy. For the protection of quantum data, this is the first quantum protocol that can be used against the most general adversaries. Furthermore, we show how the robustness in the classical case can be sensitive to the details of the classification model, but in the quantum case the details of classification model are absent, thus also providing a potential quantum advantage for classical data that is independent of quantum speedups. This opens the opportunity to explore other ways in which quantum noise can be used in our favour, as well as identifying other ways quantum algorithms can be helpful that is independent of quantum speedups. 

\end{abstract}
\maketitle 
\section{Introduction}
Noise in quantum information processing has long been viewed as a feature to avoid and remove, notably in quantum computation. However, in the Noisy Intermediate-Scale Quantum (NISQ) era of near-term quantum computing~\cite{preskill2018quantum}, the presence of noise is inevitable. The focus is both on reducing the effects of quantum noise, for example using error-mitigation~\cite{endo2018practical,temme2017error} and for finding protocols whose integrity can nevertheless withstand this noise. However, a parallel approach can be taken to instead study noise under a positive lens. In classical information processing, noise is actively leveraged in many applications including strengthening security and privacy using differential privacy~\cite{dwork2011differential}, enhancing weak signals using stochastic resonance~\cite{gammaitoni1998stochastic}, improving signal resolution after truncating data with dithering~\cite{roberts1962picture} and speeding convergence rates in neural networks~\cite{jim1995effects}. Can we look at quantum noise in this same positive light and use it to our advantage? 

One important proposed application of these quantum devices is performing machine-learning tasks like classification ~\cite{biamonte2017quantum, grant2018hierarchical} and classification algorithms can be less vulnerable against noise.  An intuitive reason behind this is that classification only has few possible outputs and machine learning can still provide accurate classification in the classical world despite the `messiness' of real-life data like images and sound recordings. Indeed, a recent work~\cite{larose2020robust} showed how quantum binary classifiers can be made robust against common sources of quantum noise by choosing a right encoding of classical data into quantum states.  

However, despite being tolerant to small amounts of noise with known sources, classification algorithms are generally not protected against unknown `worst-case' noise sources, such as adversarial attacks. In fact, classification algorithms in machine learning are often very sensitive to adversarial attacks and this presents a key obstacle for the future development of classical machine learning~\cite{szegedy2013intriguing}. These adversaries perturb the original data point by only a small undetectable amount, yet the new datapoint, known as an adversarial example, is completely misclassified in otherwise extremely accurate classifiers. This observation presents an impetus for the vibrant field called adversarial machine learning~\cite{huang2011adversarial, kurakin2016adversarial} and this has recently been extended to the quantum domain in adversarial quantum learning~\cite{wiebe2018hardening, liu2019vulnerability,lu2019quantum}. While many important methods focus on finding new and more robust versions of existing algorithms~\cite{goodfellow2018making}, including on quantum devices~\cite{wiebe2018hardening, lu2019quantum}, this approach is generally vulnerable to counterattacks and don't provide theoretical guarantees against all possible adversaries~\cite{yuan2019adversarial}. 

We take a different approach that does not require inventing new algorithms to improve robustness, yet can provide a robustness guarantee against any unknown perturbation, such as from an adversary. We begin from our intuition that noise is a kind of scrambling mechanism. It can `scramble' the effects of disturbances made to one's original data, for instance by adversaries, thus diminishing the effects adversarial attacks can have. Therefore we can ask whether noise, instead of hindering the computation, can in fact assist in the presence of adversarial attacks? 

More specifically, noise in the classical realm has been associated with improving the privacy of algorithms, providing a property called differential privacy~\cite{dwork2011differential}. Differential privacy is the property of an algorithm whose output cannot distinguish small changes in the initial dataset, like the presence or absence of one party's datapoint, hence in this way preserving privacy of that party. This is in fact the very property we want in making our algorithm robust against adversarial examples, which are small changes to the initial dataset that induce misclassification. 

We demonstrate that by including depolarisation in one's quantum circuit for classification, we can achieve quantum differential privacy and in turn, be able to provide robustness bounds in the presence of adversaries which were not possible before. This is the most natural mechanism to exploit noise to protect quantum data, which appear in condensed matter systems, quantum communication networks, quantum simulation, quantum metrology and quantum control. In addition, we show how the robustness bound in the classical case can be sensitive to the details of the classification model but in the quantum case this bound is dependent only on the number of possible class categories and no other feature of the classification model. This therefore demonstrates an important example of a security advantage in performing a classification algorithm on a quantum device versus a purely classical device, for both quantum and classical data. 

We begin by defining classification, adversarial examples and differential privacy. Then we demonstrate how adding depolarisation noise in quantum classifiers can induce quantum differential privacy which can in turn provide protection against adversarial examples. 
\section{Background}\label{sec:bac}
We briefly review the classification problem in both the classical and quantum domains before introducing the concept of adversarial examples. We then define classical and quantum differential privacy, which we later employ as a key tool to achieve robustness of our classifier against adversarial examples.
\subsection{Classification task} \label{subsec:classification}
A classification task is a mapping from a set of classical or quantum input states to a label chosen from a finite set.  If the size of this finite set is $K\geq 2$, we have a $K$-multiclass classification problem  \cite{goodfellow2016deep}. $K=2$ is the special case of binary classification, e.g., given images of only ants or cicadas, to decide which picture belongs to which insect. 
\begin{definition}[$K$-multiclass classification]\label{def:K_multi}
The algorithm $\mathcal{A}:\Sigma \rightarrow \mathcal{C}$ is called a $K$-multiclass classification algorithm if it maps the set of input states $\Sigma$ onto the set $\mathcal{C}=\{0,...,K-1\}$. Let the state $\sigma \in \Sigma$ and $C\in \mathcal{C}$. If $\mathcal{A}(\sigma)=C$, then $C$ is the predicted class label assigned to $\sigma$.
\end{definition}
In machine learning, the algorithm $\mathcal{A}$ does not need to be pre-defined and can instead be learned through a training dataset $\mathcal{D}$. This dataset $\mathcal{D}=\{\sigma_i, \vec Y (\sigma_i)\}_{i=1}^M$ consists of $M$ pairs of input states $\sigma_i$ and their corresponding class labels represented by the $K$-dimensional vector $\vec Y (\sigma_i)$. Its $k^{\text{th}}$ entry $\vec Y_k(\sigma_i)=1$ if the class label of $\sigma_i$ is $k$ and every other entry of $\vec Y_k(\sigma_i)$ is zero otherwise. To learn $\mathcal{A}$, we first define a parameterised function $f(\vec \theta, \sigma_i) \in  \mathbb{R}^K$ where $\vec \theta$ are free parameters that can be tuned. The learning happens as $\vec \theta$ is optimized to minimize the empirical risk
\begin{align}\label{eqn:clc_bina}
	\min_{\vec{\theta}}\frac{1}{M}\sum_{i=1}^M \mathcal{L}(f(\vec{\theta}, \sigma_i), \vec Y(\sigma_i)),
\end{align} 
where $\mathcal{L}$ refers to a predefined loss function. The goal in learning is to minimise this empirical risk Eq.~\eqref{eqn:clc_bina} for one's given training dataset $\mathcal{D}$, where the optimized parameters are denoted $\vec{\theta}^*$. Given test state $\sigma$, we can define $\vec y(\sigma)=f(\theta^*, \sigma)/\|f(\theta^*, \sigma)\|_1$ as the score vector among $K$ labels, where $\|\cdot\|_1$ denotes the $l_1$-norm and $\vec y(\sigma)\in\mathbb{R}^K$ is the normalized vector of $f(\theta^*, \sigma)$. Then the $k^{\text{th}}$ entry of the vector function $f(\vec{\theta}^*, \sigma)=\vec y_k(\sigma)\in [0,1]$ can be interpreted as the probability that $\sigma$ is assigned the label $k$. Then the learned classification algorithm $\mathcal{A}$ outputs the class label $C$ for a input state $\sigma$ using the condition
\begin{align}\label{eqn:pred}
	C \equiv \arg\max_k  \vec y_k(\sigma) \equiv \mathcal{A}(\sigma), 
\end{align}
where the final class label $C$ is decided by identifying the class label with the highest corresponding probability.

For the quantum $K$-multiclass classification task with quantum test state $\sigma$ we can employ a quantum circuit, see Fig.~\ref{fig:Fig1}(a), to compute $\vec y(\sigma)$ instead of using a classical circuit. We can identify $\vec y_k(\sigma)$ to be the probability of the final measurement outcome of the quantum circuit being $k$, 
\begin{align} \label{eq:yk1}
	\vec y_k(\sigma) = \Tr(\Pi_k\mathcal{E}(\sigma \otimes |a\rangle \langle a|)),
\end{align}
where $\Pi_k$ is a POVM, $\mathcal{E}$ is a quantum operation that contains information about the trained parameters $\vec \theta^*$ \cite{benedetti2019parameterized} and $|a\rangle \langle a|$ is an ancilla.
\begin{figure}[h!]
\centering
\includegraphics[width=0.45\textwidth]{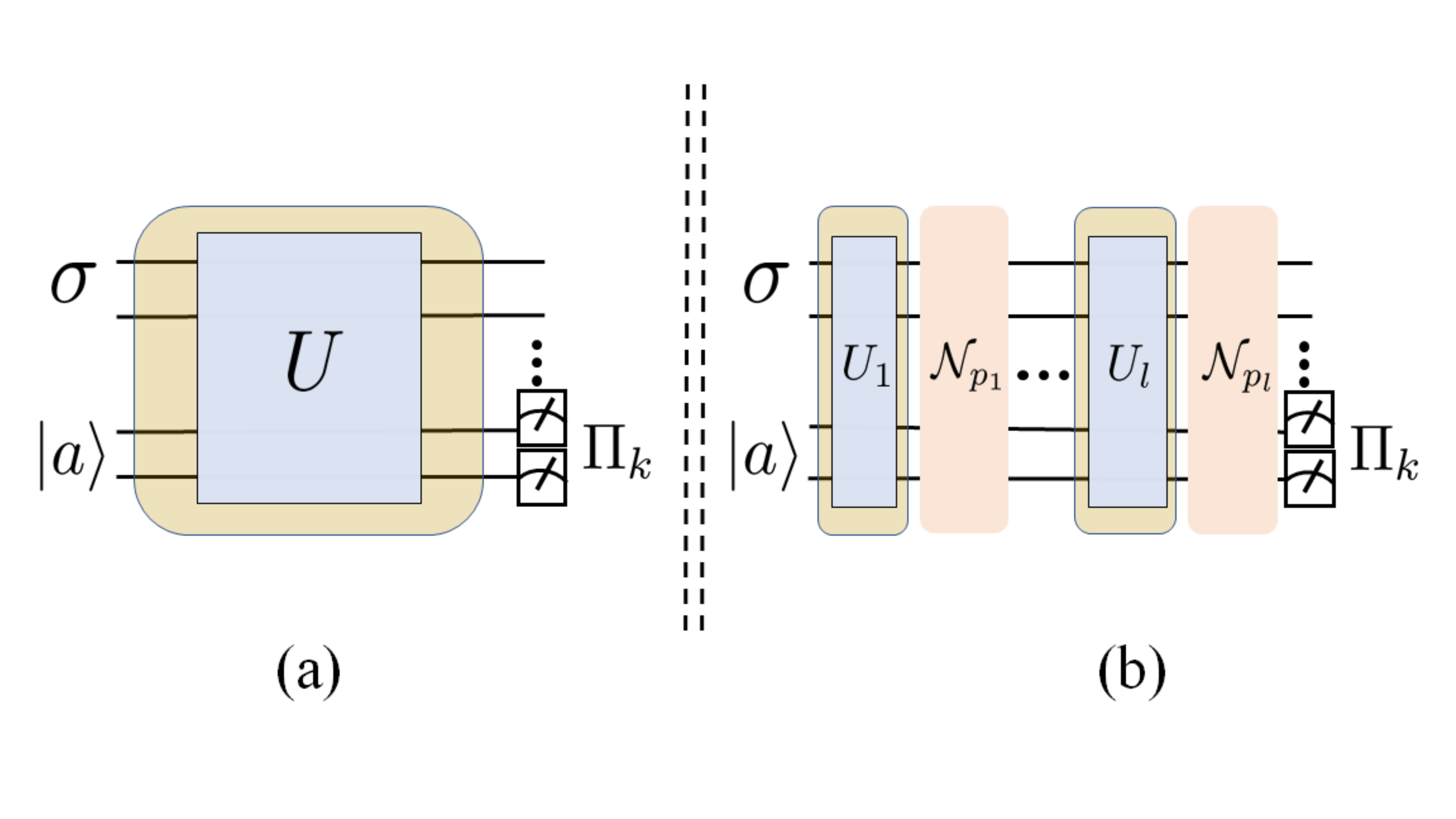}
\caption{\label{fig:Fig1} (a) A generic quantum circuit to estimate $\vec y_k(\sigma)$, which is the probability that test state $\sigma$ is assigned a class label $k$ in a $K$-multiclass classification problem. $|a\rangle$ is an ancilla state where $\sigma \otimes |a\rangle \langle a|$ is $D$-dimensional and $\Pi_k$ is $D_{meas}$-dimensional, where $D_{meas}\geq K$. With finite $N$ measurements at the output, one obtains an estimate $\vec y^{(N)}_k(\sigma)$ for $\vec y_k(\sigma)$. (b) Adding depolarisation noise channels $\mathcal{N}_{p_i}$ along the circuit, where $i=1,...,l$, the output in the $N \rightarrow \infty$ sampling limit becomes $\tilde{\vec y}_k(\sigma)$. With finite $N$ measurements at the output, one obtains the estimate $\tilde{\vec y}^{(N)}_k(\sigma)$. See text for details.}
\end{figure}
However, precise values of the probabilities $\vec y_k(\sigma)$ can only be obtained in the infinite sampling regime. This means that if only $N$ measurements are allowed at the output of the circuit, we can only obtain an estimated value $\vec y^{(N)}_k(\sigma)$ of the output probabilities. 
\subsection{Adversarial examples}
Adversarial examples are attacks on input examples to classification problems that lead to misclassification. In particular, these include worst-case attacks where the adversary can craft small imperceptible perturbations $\sigma \rightarrow \rho$ about a given correctly classified input $\sigma$ that result in misclassification \cite{goodfellow2014explaining}. This means that while the true labels $\sigma$ and $\rho$ are identical, if $\rho$ is an adversarial example, $\mathcal{A}$ will class them differently. We can define adversarial examples more formally as follows \cite{sharif2018suitability}. 
\begin{definition}[Adversarial example]
Suppose we are given a well-trained classification function $\mathcal{A}(\cdot)$ as defined in Eq.~\eqref{eqn:pred}, an input example $(\sigma, C)$, a distance metric $h(\cdot,\cdot)$ and a small enough threshold value $L$. Then $\rho$ is said to be an adversarial example if the following is true
	\begin{align}
		(\mathcal{A}(\sigma)=C) \wedge   (\mathcal{A}(\rho ) \neq C)  \land (h(\sigma ,\rho)\leq L).
	\end{align}
\end{definition} 
If $\sigma, \rho$ are classical states, suitable distance metrics are the $l_p$-norms, so $h(\sigma,\rho)=||\sigma, \rho||_p$. If $\sigma, \rho$ are quantum states, we will use the trace distance $h(\sigma,\rho)=\tau(\sigma,\rho)=\text{Tr}(|\rho-\sigma|)/2$. 

In the rest of this paper, we will use Greek letters to refer to quantum states and bold Roman letters to refer to classical states unless otherwise specified.
\subsection{Differential privacy}
Differential privacy is an important concept in computer science that quantifies the sensitivity of the outputs of algorithms to changes in their input data. The less sensitive it is, the better the algorithm can preserve the privacy of the input data. Here we can formulate the definition of classical differential privacy as follows \cite{dwork2011differential}.
\theoremstyle{definition}
\begin{definition}[Classical differential privacy]
Suppose $\mathcal{M}$ is a classical algorithm that takes as input entries $\vec x \in X$ of some classical database $X$ and outputs values belonging to the set $\mathcal{S}$. Then $\mathcal{M}$ is said to satisfy \textit{classical $(\epsilon,\delta)$-differential privacy} if, for all $\vec x \in X$, $\vec x' \in X'$ which are separated by a small distance, e.g., Hamming distance $h(\vec x, \vec x')\leq 1$ and all measurable sets $\mathcal{S}\subseteq \text{Range}(\mathcal{M})$, 
\begin{align} \label{eq:cldp}
Pr(\mathcal{M}(\vec x) \in \mathcal{S}) \leq e^{\epsilon} Pr(\mathcal{M}(\vec x') \in \mathcal{S})+\delta,
\end{align}
where $Pr(\cdot)$ denotes the probability of $(\cdot)$ and $\epsilon, \delta>0$. We call $(\epsilon, \delta)$ the \textit{privacy budget} for the algorithm. 
\end{definition}
Informally, this definition says that for two input data points separated by a small distance, a small privacy budget means that the output of the algorithm differs very little, hence the input information is partially kept private. The selection of this distance $h(\cdot, \cdot)$ varies depending on the task, e.g., Hamming distance or $l_p$ distance \cite{dwork2011differential}. A natural distance $h(\cdot, \cdot)$  for quantum data is the trace distance, which we can employ in a definition for quantum differential privacy \cite{zhou2017differential} which we will use throughout this paper. An alternative definition for quantum differential privacy \cite{aaronson2019gentle} does not require quantum data $\sigma$ and $\tau$ to be close in trace distance, but rather that $\rho$ is obtainable by applying a quantum operation on only a single register of $\sigma$. See also \cite{arunachalam2020quantum} for a related definition applied to PAC learning. However, for our purposes of working directly with quantum states $\sigma$ and $\rho$, the use of trace distance is the most appropriate. 

Suppose $\mathcal{M}(\sigma, \Pi_{\mathcal{S}})$ is a quantum algorithm that takes input state $\sigma$, applies a quantum operation $\mathcal{E}$ before applying the POVM $\{\Pi_{k}\}$, where the set of final measurement results $k \in \mathcal{S}$. These set of outcomes are then observed with probability $\Pr(\mathcal{M}(\sigma,\Pi_{\mathcal{S}}) \in  \mathcal{S})= \sum_{k\in \mathcal{S}}\Tr(\Pi_k\mathcal{E}(\mathcal{\sigma}))$. By analogy with Definition \ref{def:K_multi}, we can write a definition of quantum differential privacy following Zhou and Ying \cite{zhou2017differential}.
 \theoremstyle{definition}
\begin{definition}[Quantum differential privacy]\label{def:QDP}
The quantum algorithm $\mathcal{M}$ satisfies \textit{$(\epsilon,\delta)$-quantum differential privacy} if for all input quantum states $\sigma$ and $\rho$ with $\tau(\sigma, \rho)<\tau_D$ and for all measurable sets $\mathcal{S}\subseteq \text{Range}(\mathcal{M})$ (equivalently, for every $\Pi_S \subseteq  \{\Pi_k\}_{k=0}^{\max k=K-1}$) 
\begin{align}
\Pr(\mathcal{M}(\rho, \Pi_{\mathcal{S}}) \in\mathcal{S}) \leq e^{\epsilon} \Pr(\mathcal{M}(\sigma, \Pi_{\mathcal{S}}) \in \mathcal{S})+\delta.
\end{align}
\end{definition}
For the rest of the paper, we focus on the case $\delta=0$, which is referred to as \textit{$\epsilon$-quantum differential privacy}. To illustrate a simple example, suppose we have a binary classification problem where we choose the POVM $\{\Pi_0, \Pi_1=\mathbf{1}-\Pi_0\}$. The probability $\sigma$ is assigned class labels $k=0,1$ by a quantum binary classifier is $\tilde{\vec y}_0(\sigma) \equiv \text{Tr}(\Pi_0(\mathcal{E}(\sigma))$ and $\tilde{\vec y}_1(\sigma)=1-\tilde{\vec y}_0(\sigma)$ respectively. Then if $\mathcal{M}$ satisfies $\epsilon$-quantum differential privacy, Definition \ref{def:QDP} requires that we must satisfy
\begin{align} \label{eq:def4example}
e^{-\epsilon}\leq \frac{\tilde{\vec y}_k(\rho)}{\tilde{\vec y}_k(\sigma)}\leq e^{\epsilon}.
\end{align}
\section{Improving robustness of quantum classifiers against adversaries by adding noise}
In this section, we show how the presence of depolarisation noise in quantum circuits for classification improves robustness against adversarial examples. We begin with our definition of adversarial robustness. 
\begin{definition}[Adversarial robustness]
Let the test state $\sigma$ have the class label $\mathcal{A}(\sigma)$ under a classification algorithm $\mathcal{A}$. Then $\mathcal{A}$ is said to possess \textit{adversarial robustness of size $\tau_D$} if for all $\sigma$ that is perturbed $\sigma \rightarrow \rho$ by an unknown source where $\tau(\sigma,\rho)\leq \tau_D$, the class label of $\rho$ does not change, i.e., $\mathcal{A}(\rho)=\mathcal{A}(\sigma)$.  
\end{definition}
We must emphasise here the difference between robustness bounds against a known noise source versus an unknown adversary. Protection against an unknown adversary is a robustness guarantee against a worst-case scenario, whereas commonly-appearing known noise sources are usually far from the worst-case scenario. 

Our goal is to demonstrate how a naturally-occurring known noise source can be used to protect a quantum classifier against worst-case adversarial perturbations. This can be done in three main steps. We first show the robustness of quantum classifiers to this known noise source, then demonstrate how this gives rise to quantum differential privacy for the classifier. Finally we prove how quantum differential privacy can be used to derive a theoretical bound against general adversaries. 

One such naturally-occurring quantum noise source is the depolarisation noise channel $\mathcal{N}_{p}$, which acts on a $D$-dimensional state $\Sigma_D$ like 
\begin{align}\label{eq:noisedep}
\mathcal{N}_{p}(\Sigma_D)=p \frac{\mathbb{I}_{D}}{D}+(1-p)\Sigma_D,
\end{align}
where $\mathbb{I}_{D}$ is the $D \times D$ identity matrix and $p\in [0,1]$. Before the final measurement, we can represent our quantum classifier as a unitary $U$ gate acting on an input state $\sigma\otimes |a\rangle \langle a|$, as represented in Fig.~\ref{fig:Fig1}(a). We can then add $\mathcal{N}_{p_i}$ after each unitary $U_i$ where $U=U_1...U_l$ and $i=1,...,l$. Here $l$ is the total number of depolarisation channels with noise parameters $p_i>0$. This noisy circuit is depicted in Fig.~\ref{fig:Fig1}(b). The output of this noisy $K$-multiclass classification circuit given test state $\sigma$ can be written as 
\begin{align} \label{eq:tildeyk1}
\tilde{\vec y}_k (\sigma)\equiv \text{Tr}(\Pi_k \mathcal{N}_{p_l}(U_l(...\mathcal{N}_{p_1}(U_1(\sigma \otimes |a\rangle \langle a|)U_1^{\dagger})...)U_l^{\dagger})),
\end{align}
where it can be shown \footnote{This is an extension from Theorem 2 in \cite{larose2020robust} to beyond $K=2$ and follows by an inductive application of Eq.~\eqref{eq:tildeyk1}.} that for $p \equiv 1-\prod_{i=1}^l (1-p_i)$
\begin{align} \label{eq:ytildep}
\tilde{\vec y}_k (\sigma)=\frac{p}{K}+(1-p)\vec y_k(\sigma). 
\end{align}
This leads to the interesting observation that the noisy test score $\tilde{\vec y}_k (\sigma)$ is independent of where depolarisation channels are placed in the circuit. Furthermore, the effect of all depolarisation channels with parameters $p_i$ can be replaced by a single depolarisation channel with parameter $p \equiv 1-\prod_{i=1}^l (1-p_i)$. In the trivial case $p_i=0$ for all $i$, $p=0$. For the rest of this paper, we will for simplicity replace the effect of all noise parameters $p_i$ with $p$ unless stated otherwise. 

Before achieveing our goal, we first need Eq.~\eqref{eq:ytildep} to prove the following lemma showing that the $K$-multiclass classification algorithm performed by the noisy circuit is robust against depolarisation noise for any $0\leq p_i<1$. This is a generalisation of a recent result from LaRose and Coyle \footnote{This appears in Theorem 2 in \cite{larose2020robust}. Also see \cite{larose2020robust} for a list of common types of noise that binary quantum classifers are naturally robust against as well as interesting encoding strategies to induce robustness when the classifiers are not naturally robust.} to the case of $K$-multiclass classification.
\begin{lemma}\label{lem:K-multi-classification}
Let $\vec y_k(\sigma)$ denote the output for the noiseless circuit in Fig.~\ref{fig:Fig1}(a), i.e., $p_i=0$ for all $i$. Then if the class label $C$ is assigned to $\sigma$ by the noiseless circuit, i.e., $C=\arg \max_k \vec y_k(\sigma)$, then the same label is also assigned by the noisy circuit, which has $p_i>0$ for at least one $i$. This means $\arg \max_k \tilde{\vec y}_k(\sigma)=C$ for any $\sigma$ and $0\leq p_i<1$. Furthermore, if $\arg \max_k \tilde{\vec y}_k(\sigma)=C$ then $C=\arg \max_k \vec y_k(\sigma)$. 
\end{lemma}
\begin{proof}[Proof of Lemma \ref{lem:K-multi-classification}]
For details please see Appendix~\ref{sec:lemma1}.
\end{proof} 
The above result demonstrates robustness of quantum classifiers against depolarisation noise if one has access to the exact probabilities $\tilde{\vec y}_k(\sigma)$. However, this is only possible in the limit of infinite sampling. If one is only able to sample the circuit $N$ times, one instead obtains only the estimated values $\tilde{\vec y}_k^{(N)}(\sigma)$. Then to guarantee robustness against depolarisation noise to high probability, we find the following required sampling complexity $N$ increases only with increasing depolarisation noise $p$, but is not dependent on the dimensionality of $\sigma$. 
\begin{prop}\label{prop:prop1}
	Let the predicted classification label of $\sigma$ using the noiseless $K$-multiclass classification circuit be $C$. This means we can define $\xi \equiv \vec y_C(\sigma)-\max_{k \neq C}\vec y_k(\sigma)$ where $\xi>0$. In the corresponding circuit with depolarisation noise parameters $p_1,...,p_l$, one samples the circuit $N$ times for each $k$ to obtain the estimates $\tilde{\vec{y}}_k^{(N)}(\sigma)$. Then $\sigma$ is also labelled $C$ with probability at least $\beta$ if the sample complexity $N\sim 1/(8\xi^2(1-p)^2)\ln(2/(1-\beta))$, where $p \equiv 1-\prod_{i=1}^l (1-p_i)$. 
\end{prop}
\begin{proof}[Proof of Proposition \ref{prop:prop1}]
	A basic sketch of the proof is the following. It can be shown that $\eta \equiv \tilde{\vec {y}}_C(\sigma)-\max_{k\neq C}\tilde{\vec{y}}_k(\sigma)=p\xi$. Thus one requires sufficient $N$ to resolve the difference $\tilde{\vec{y}}_C^{(N)}(\sigma)-\tilde{\vec{y}}_k^{(N)}(\sigma)$ to within $2\eta$. We then employ Hoeffding's inequality \cite{mohri2018foundations} to bound the sample complexity. Please see Appendix~\ref{sec:proposition1} for details. 
\end{proof} 
 
Now we show how adding depolarisation noise gives rise to quantum differential privacy for our algorithm. This is an application of a result from Zhou and Ying~\cite{zhou2017differential} for our quantum classifier. 

\begin{lemma}\label{lem:DP-robust}
	Let the algorithm $\mathcal{M}$ correspond to the $K$-multiclass classification circuit defined in Fig.~\ref{fig:Fig1}(b) with depolarisation noise channels $\mathcal{N}_{p_i}$, where $i=1,...,l$ and $p \equiv 1-\prod_{i=1}^l (1-p_i)$, and measurement operators $\{\Pi_k\}_{k=1}^K$. Then for two quantum test states $\sigma$ and $\rho$ obeying $\tau(\sigma,\rho)\leq \tau_D$ with $0\leq \tau_D\leq 1$, $\mathcal{M}$ satisfies $\epsilon$-quantum differential privacy where
\begin{align}\label{eq:dpdefdep}
\epsilon=\ln\left(1+D_{meas}\frac{(1-p)\tau_D}{p}\right) 
\end{align}
and $D_{meas}\geq K$ is the dimension of the operators $\{\Pi_k\}_{k=1}^K$.
\end{lemma}
\begin{proof}[Proof of Lemma \ref{lem:DP-robust}]
	 This is equivalent to Theorem 3  from \cite{zhou2017differential} applied to our quantum classifier, but we extend to the case where we can apply multiple depolarisation channels $\mathcal{N}_{p_i}$. For details please see Appendix~\ref{sec:lemma2proof}.  
\end{proof}

Lemma \ref{lem:DP-robust} states that the privacy budget $\epsilon$ in the presence of depolarisation noise decreases with increasing $p \equiv 1-\prod_{i=1}^l (1-p_i)$, hence higher depolarisation noise parameters gives greater differential privacy. Furthermore, this privacy is independent of where one inserts depolarisation noise because the product $\prod_{i=1}^l (1-p_i)$ is invariant under permutation of its factors. It is also independent of any details of the classifier except $D_{meas}$, which serves as an upper-bound to the number of class labels in our classifier. We will return to these points later. 

Using the results of Lemmas \ref{lem:K-multi-classification} and \ref{lem:DP-robust}, the following theorem demonstrates that by increasing the strength of depolarisation noise in our circuit, this also increases our $K$-multiclass classifier's robustness against adversarial examples. 

\begin{theorem}[Infinite sampling case] \label{thm:DP_infini}
	We begin with our $K$-multiclass classification circuit with depolarisation noise parameters $p_i$ where $i=1,...,l$ and $p \equiv 1-\prod_{i=1}^l (1-p_i)$. Let infinite sampling of the output be allowed, so we can find $\tilde{\vec y}_k(\rho)$ for $k=0,...,K-1$ for any test state $\rho$ given.  Suppose $\tilde{\vec{y}}_C(\sigma)>e^{2\epsilon}\max_{k \neq C}\tilde{\vec{y}}_k (\sigma)$ holds, where $\epsilon=\ln(1+D_{meas}(1-p)\tau_D/p)$, which implies that $\sigma$ is assigned the class label $C$, i.e., $C=\arg \max_k \tilde{\vec{y}}_k(\sigma)=\arg \max_k \vec y_k(\sigma)$. Then $\rho$ is also labelled as $C$, i.e., $C=\arg \max_k \tilde{\vec y}_k (\rho)=\arg \max_k \vec y_k (\rho)$ for any $\rho$ where $\tau(\sigma,\rho)\leq \tau_D$.
\end{theorem}

\begin{proof}[Proof of Theorem \ref{thm:DP_infini}]
	Please refer to Appendix~\ref{sec:appendixc} for the proof.
\end{proof}
 
This means that if a test state $\sigma$ undergoes an arbitrary adversarial perturbation $\sigma \rightarrow \rho$, the classification of $\rho$ will remain identical to that of $\sigma$ for a larger range of $\tau(\sigma,\rho)$ if $p$ increases. Furthermore, if $\tau_D$ remains constant, then the extra condition required of the input state $\tilde{\vec y}_C(\sigma)>e^{2\epsilon}\max_{j \neq C}\tilde{\vec y}_j (\sigma)$ also becomes easier to satisfy as $p$ increases. A similar result holds for the finite sampling case. 
\begin{theorem}[Finite sampling case]\label{thm:DP_fini}
	Suppose one samples the output of the circuit $N$ times for the estimation of each $\tilde{\vec{y}}_k(\sigma)$. Let $\tilde{\vec{y}}_C^{(N)}(\sigma)-\zeta>e^{2\epsilon}\max_{k \neq C}(\tilde{\vec{y}}_k^{(N)}(\sigma)+\zeta)$ where $\epsilon=\ln(1+D_{meas}(1-p)\tau_D/p)$, which implies $\sigma$ has the class label $C$. Then the class label of $\rho$ is also $C$, i.e., $C=\arg \max_k \vec{y}_k(\rho)=\arg \max_k \tilde{\vec{y}}_k(\rho)$ to probability at least $1-2\exp(-2N\zeta^2)$ for any $\rho$ where $\tau(\sigma,\rho)\leq \tau_D$. This also implies $\tilde{\vec{y}}_C^{(N)}(\rho)+\zeta>\max_{k\neq C}\tilde{\vec{y}}_k^{(N)}(\rho)-\zeta$ to probability at least $1-2\exp(-2N\zeta^2)$.
\end{theorem}
\begin{proof}[Proof of Theorem \ref{thm:DP_fini}]
 	We employ Hoeffding's inequality \cite{mohri2018foundations} to show $\tilde{\vec{y}}_k^{(N)}(\sigma)-\zeta \leq \tilde{\vec{y}}_k(\sigma)\leq \tilde{\vec{y}}_k^{(N)}(\sigma)+\zeta$ is true to probability at least $1-2\exp(-2N\zeta^2)$. This relates the finitely sampled estimates $\tilde{\vec{y}}_k^{(N)}(\sigma)$ to $\tilde{\vec{y}}_k(\sigma)$ from infinite sampling. Then we can apply the results of Theorem \ref{thm:DP_infini} for infinite sampling to prove our results. Please see Appendix~\ref{sec:appendixd} for details of the proof.
\end{proof}
As special examples, we now explore the robustness property of two discriminative learning models for binary classification: quantum neural network and quantum kernel classifiers. 
\subsection{Quantum neural network}
\begin{figure}
\centering
\includegraphics[width=0.45\textwidth]{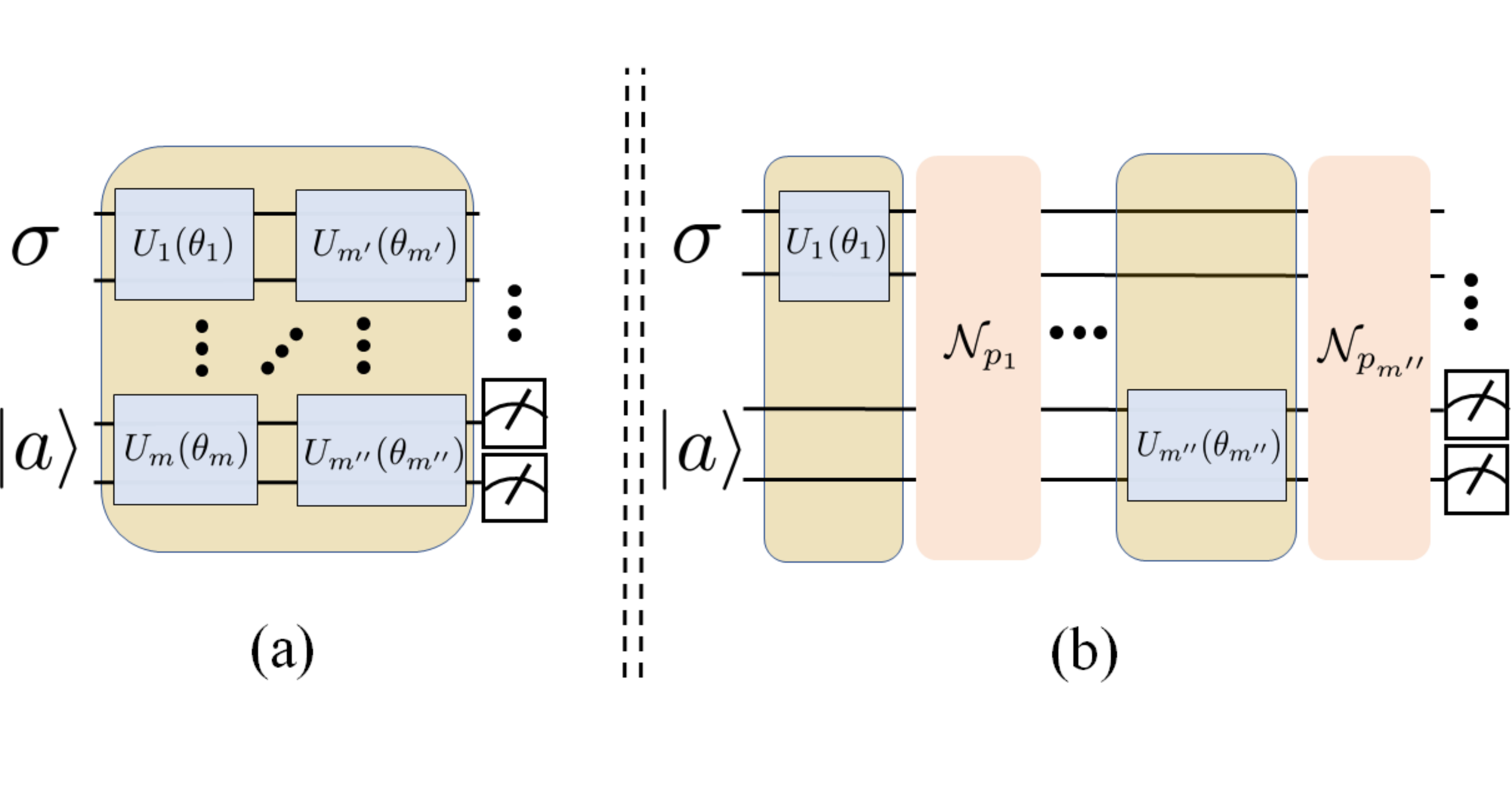}
\caption{\label{fig:noise1} 
\textit{Noiseless and noisy QNN circuits.} (a) The basic scheme of QNN (noiseless). The trainable unitary $U(\bm{\theta})$ (yellow region) is composed of the product of parameterised single-qubit gates and fixed two-qubit gates $U_i(\theta_i)$ where $i=1,...,m''$ and in the diagram above $1\leq m\leq n \leq m'\leq m''\leq nl$ where $l$ is the depth of the circuit and $n=\log_2 D$. The test state is $\sigma$ and the ancilla state is $|a\rangle$. (b) Our protocol for QNN (noisy) where the depolarisation channels $\mathcal{N}_{p_i}$ (pink region) are added to the noiseless QNN circuit.}
\end{figure} 
The quantum neural network (QNN), proposed by \cite{farhi2018classification}, is a building block for various quantum learning models \cite{schuld2018circuit,huggins2018towards,havlivcek2019supervised,farhi2018classification,benedetti2018generative,dallaire2018quantum}. The basic scheme of QNN is illustrated in Figure~\ref{fig:noise1} (a), which is a special case of the circuit in Fig.~\ref{fig:Fig1}(a). The $D$-dimensional quantum input state is $\sigma \otimes |a\rangle \langle a|$, where $\sigma$ refers to either the training or test states and $|a \rangle$ is an ancilla. The trainable unitary $U(\bm{\theta})\in \mathbb{C}^{D\times D}$ is then applied, which consists of trainable single-qubit gates and fixed two-qubit gates. Our protocol for QNN, as shown in Figure \ref{fig:noise1} (b), employs the depolarisation channels $\mathcal{N}_{p_i}$ that can appear within the QNN circuit before final measurements with POVM $\{\Pi_k\}$. 

The typical application of QNN is for binary classification, broadly used in \cite{farhi2018classification,huggins2018towards,benedetti2018generative, dallaire2018quantum}, where one makes single-qubit measurements using $\{\Pi_0,\Pi_1=\mathbf{1}-\Pi_0\}$ and $D_{meas}=2$. We can apply Theorem \ref{thm:DP_infini} directly to our scenario and we have the following corollary. 

\begin{corollary}
 \label{coro:DP_infi}
	Let the given input $\sigma$ be given the classification label `0'  and define $\tilde{\vec y}_0(\sigma)/\tilde{\vec y}_1(\sigma)\equiv B$.  In binary classification, QNN, with depolarisation channels $\mathcal{N}_{p_i}$ and $p\equiv 1- \prod_{i=1}^l (1-p_i)$, is robust against any perturbations $\sigma \rightarrow \rho$ with $\tau(\sigma,\rho)<\tau_D$ and $\epsilon=\ln(1+2(1-p)\tau_D/p)$, if 
	\begin{equation}\label{eqn:B_infi}
		B>\exp(2\epsilon)~.
	\end{equation}
\end{corollary}	
Since $D_{meas}=2$ for binary classification, we note that the privacy budget $\epsilon$ is now \textit{independent} of the dimension of the problem. Therefore, even as the feature dimension of the input $\sigma$ grows, it does not affect the robustness of the classifier against adversarial examples so long as some depolarisation noise with $0<p<1$ has been added to the circuit. This independence is an interesting contrast to the result in \cite{liu2019vulnerability} which states that robustness should decrease as dimensionality of $\sigma$ grows. This contradiction is resolved by observing that, unlike in \cite{liu2019vulnerability} which places no constraints on distribution from which the input states $\sigma$ are selected, here we have Eq.~\eqref{eqn:B_infi} which imposes a constraint. 

In the finite sampling limit, we can employ Theorem \ref{thm:DP_fini} to apply to our binary classifier and we have the following corollary. 

\begin{corollary} \label{coro:DP_fini}
	Let the input $\sigma$ be given the classification label `0'  and define $(\tilde{\vec y}^{(N)}_0(\sigma)-\zeta)/(\tilde{\vec y}^{(N)}_1(\sigma)+\zeta)\equiv B$, where the probabilities are estimated using $N$ samples of the quantum circuit. Then if 
	\begin{equation}\label{eqn:B_fini}
		B>\exp(2\epsilon)~,
	\end{equation}
the binary classification performed by the QNN circuit, with depolarisation channels $\mathcal{N}_{p_i}$ and $p\equiv 1- \prod_{i=1}^l (1-p_i)$, is robust to adversarial attacks $\sigma \rightarrow \rho$ with the probability at least $1-2\exp\left(-{2N\zeta^2} \right)$ where $\tau(\rho,\sigma)\leq \tau_D$  and $\epsilon=\ln(1+2(1-p)\tau_D/p)$.
\end{corollary}
\subsection{Quantum kernel classifier}
\label{sec:kernel}

\begin{figure}[h!]
\centering
\includegraphics[width=0.45\textwidth]{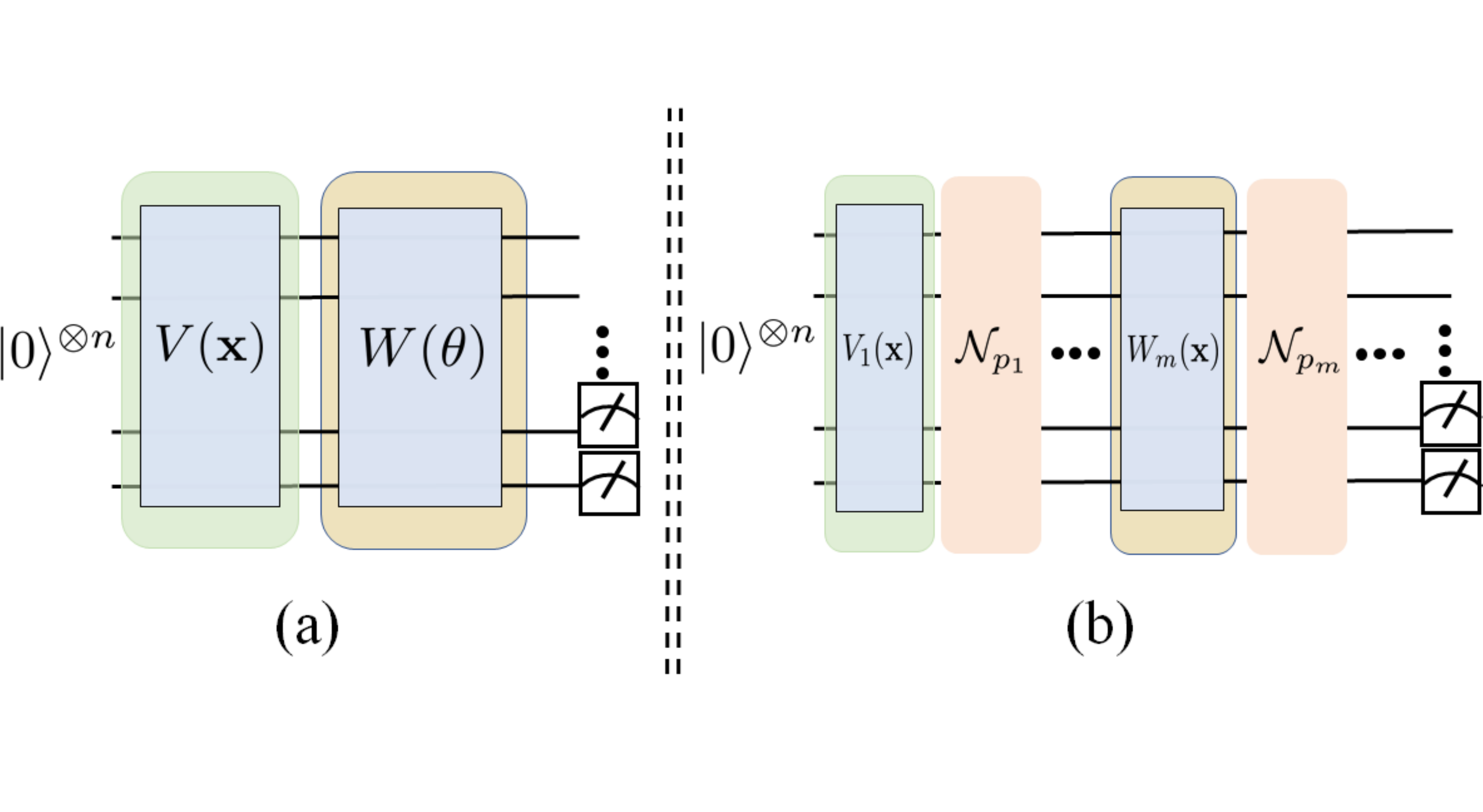}
\caption{\label{fig:noise3} \textit{Noiseless and noisy quantum kernel classifiers.} (a) A basic scheme of the quantum kernel classifier. The unitary $V(\vec x)$ (green region) takes $|0\rangle^{\otimes n}\rightarrow V(\vec x)|0\rangle^{\otimes n}$, where $n=\log_2 D$. The trainable unitary $W(\bm{\theta})$ (yellow region) is composed of trainable single-qubit gates and fixed two qubits gates, which has the same architecture as in QNN. For example, at the end we can measure in the basis $|0\rangle^{\otimes n}$ and this circuit can be used to compute the kernel $K(\vec \theta, \vec x)\equiv \langle 0|^{\otimes n}W(\bm{\theta})V(\vec x)|0\rangle^{\otimes n}$. (b) For our protocol, we can include depolarisation noise channels $\mathcal{N}_{p_i}$ (pink region) anywhere along the quantum kernel classifier.}
\end{figure}
The main idea of kernel methods is to map complex input data $\vec x$ to a higher-dimensional feature space that can then be efficiently separated \cite{goodfellow2016deep}. The generic form of a quantum kernel classifier \cite{mitarai2018quantum,havlivcek2019supervised,schuld2019quantum} is shown in Figure \ref{fig:noise3}. The output of the kernel classifier can be written as $K(\vec \theta, \vec x)\equiv\langle 0|^{\otimes n}W(\bm{\theta})V(\vec x)|0\rangle^{\otimes n}$, where $K(\vec \theta, \vec x)$ is identified with a classical kernel with test state $\vec x$ and weight vector captured by the trained $\vec \theta$ values. Here $W(\bm{\theta})$ contains the trainable parameters with the aim of minimizing the predefined loss function where the optimal occurs at $\vec \theta^*$ and $V(\vec x)|0\rangle^{\otimes n}$ refers to the kernel state that maps the input data into the higher-dimensional feature space. Thus the probability of obtaining the measurement values all `$0$' after applying $\Pi_0\equiv (|0\rangle \langle 0|)^{\otimes n}$ in the noiselss circuit is given by $\vec y_0(\vec x)=\langle 0|^{\otimes n}W(\vec \theta^*)V(\vec x)|0\rangle^{\otimes n}$. 

For a binary classification problem, the class label of $\vec x$ is $0$ if $\vec y_0(\vec x)>\vec y_1(\vec x)\equiv 1-\vec y_0(\vec x)$. In this case, $D_{meas}=D$, thus the privacy budget becomes $\epsilon=\ln(1+D(1-p)\tau_D/p)$. which grows with increasing dimensionality $D$ of the input state. Corollaries 1 and 2 then hold for the quantum kernel classifier with this modified $\epsilon$.  

\section{Numerical simulations}
We now conduct numerical simulations to illustrate our protocol for a binary QNN classifier. In particular, by leveraging the  depolarisation channel, we show how a trained QNN binary classifier has the ability to achieve certified robustness under bounded-norm adversarial attacks at testing time.  In this section, we first introduce our training dataset and the preprocessing step. We then explain the attack method that is used to evaluate the performance of our protocol. Lastly we analyse the performance of our proposed protocol. 
\subsection{Preprocessing and training procedure}
We choose to conduct our numerical simulations on the Iris dataset \cite{fisher1936use}, which has been broadly used in classical machine learning. The Iris dataset $\mathcal{D}_I=\{\sigma_i,c_i^*\}_{i=1}^{150}\in \mathbb{R}^{150\times 4} \times \mathbb{R}^{150} $ consists of three different types of Iris flowers (Setosa, Versicolour, and Virginica), where examples (belonging to Setosa) with label $c^*_i=0$ are linearly separable with respect to examples (belonging to Versicolour) with label $c^*_i=1$. 

Next, we remove all examples belonging to Virginica and denote the dataset that only contains label $c^*_i=0$ and $c^*_i=1$ as  $\mathcal{D}$, i.e., the cardinality of $\mathcal{D}$ is $100$. Then we set the fourth entry of all examples as $0$. Afterwards, we apply $l_2$  normalization to each example, i.e., $\|\sigma_i\|_2 = 1$ for any $\sigma_i\in \mathcal{D}$. Then we need to efficiently encode this classical data into quantum states \cite{schuld2017implementing}. We can then carry out the amplitude encoding method \cite{mottonen2004transformation} to encode the normalized $\sigma_i$ into a quantum state. 

Given the preprocessed dataset $\mathcal{D}$, we randomly split it into a training dataset $\mathcal{D}_{Tr}$ and a  test dataset $\mathcal{D}_{Te}$ with  $n\equiv |\mathcal{D}_{Tr}| = 60$, $|\mathcal{D}_{Te}| = 40$, and $\mathcal{D} = \mathcal{D}_{Tr}\cup  \mathcal{D}_{Te}$.  In the training procedure, we randomly sample an example   $( \sigma_i , c_i^*)$ from $ \mathcal{D}_{Tr}$ and forward $\sigma_i$ to a binary QNN classifier. For details on the circuit see Appendix~\ref{app:h}. We employ the squared loss function to train this QNN, i.e.,
\begin{equation}\label{eqn:loss_opt}
	\mathcal{L} = \frac{1}{n}\sum_{i=1}^n (c_i^*-\bar{c}_i)^2~,
\end{equation}
where $\bar{c}_i = \max_k \vec y_k(\sigma_i)\in [0,1]$ is the score vector of QNN as formulated in Subsection \ref{subsec:classification} and $\vec y_k(\sigma)$ denotes the ideal output of the QNN.   

We use the zeroth-order gradient method \cite{mitarai2018quantum} to optimize trainable parameters $\bm{\theta}$ of the QNN to minimize the loss function $\mathcal{L}$. We set the number of training epochs to $50$. The learning rate is set to $0.01$ and the total number of trainable parameters is $24$.  Figure \ref{fig:train} illustrates the training loss, training accuracy and test accuracy.  Both the training and test accuracy converges to $100\%$ after $15$ epochs (See Appendix~\ref{app:h} for more implementation details). Given the test dataset $\mathcal{D}_{Te}$, we randomly select three test examples and explore how the maximum robustness $\tau_D$ changes with varied $p$ according to Eq.~\eqref{eq:dpdefdep}, which we can rewrite as 
\begin{align} \label{eq:taudgraph}
\tau_D=\frac{(e^{\epsilon}-1)p}{D_{meas}(1-p)}.
\end{align}
Figure \ref{fig:p_tau_test} illustrates how $\tau_D$ scales with different $p$ for three different test examples with $D_{meas}=2$. Note that the constants $\epsilon$ are different for the three test examples and the test examples satisfy the condition in  Eq.~(\ref{eqn:B_infi}). In the same figure, we also plot how the test score $\tilde{\vec y}_0(\sigma)$ varies with $p$, coming from Eq.~\eqref{eq:ytildep} for the case of binary classification $K=2$ 
\begin{align}\label{eq:ykgraph}
\tilde{\vec y}_k(\sigma)=p/2+(1-p)\vec y_k(\sigma).
\end{align}  
\begin{figure} 
\centering 
\includegraphics[width=0.5\textwidth]{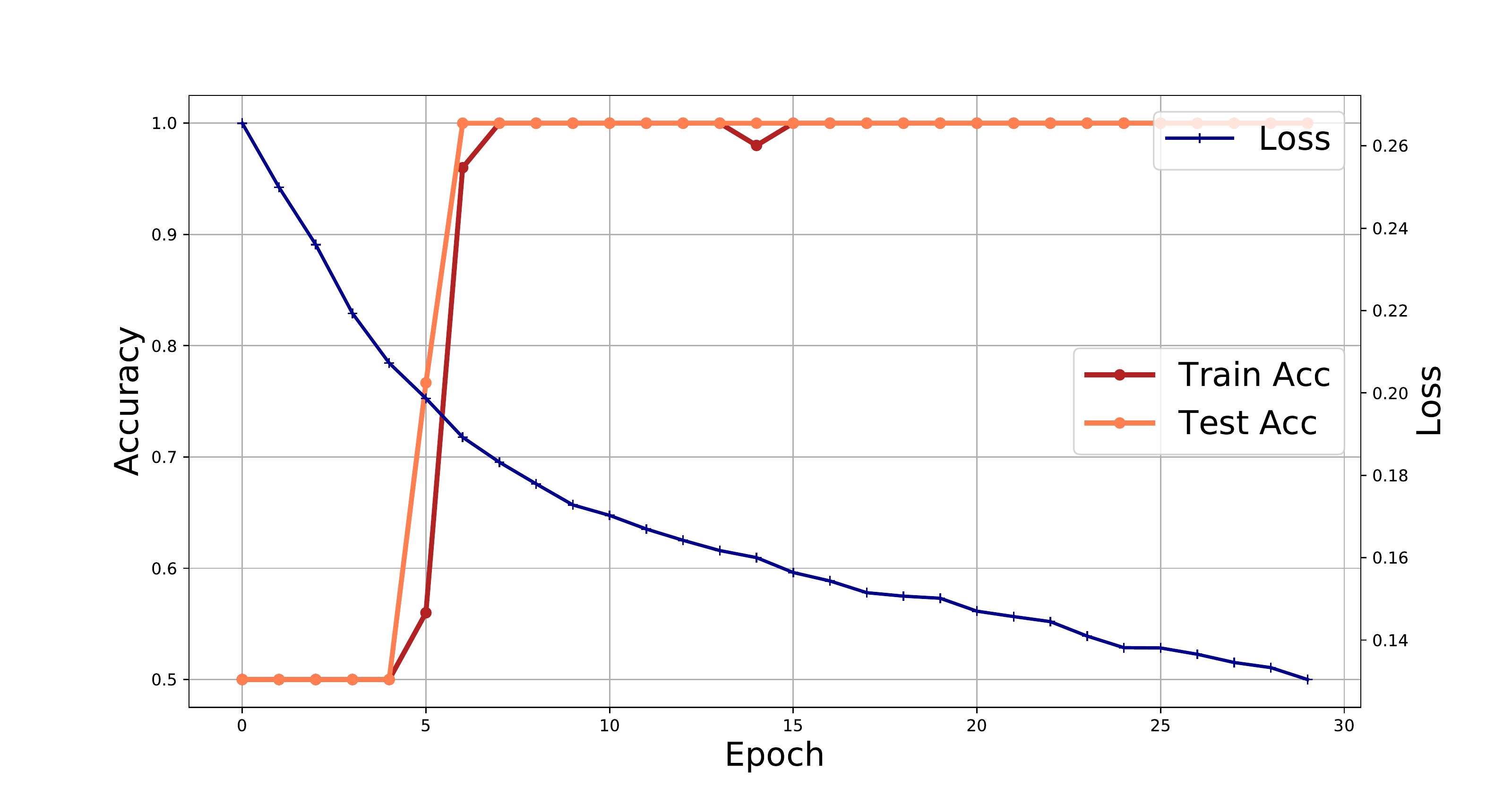} 
\caption{\small{The blue, red, and orange lines respectively show the variation of loss, the training accuracy, and the test accuracy with respect to the number of epochs. The loss continuously decreases for longer epochs, while the training accuracy and test accuracy increases and converges sharply around epoch 5.}}
\label{fig:train}  
\end{figure}  

\begin{figure}         
\includegraphics[width=0.5\textwidth]{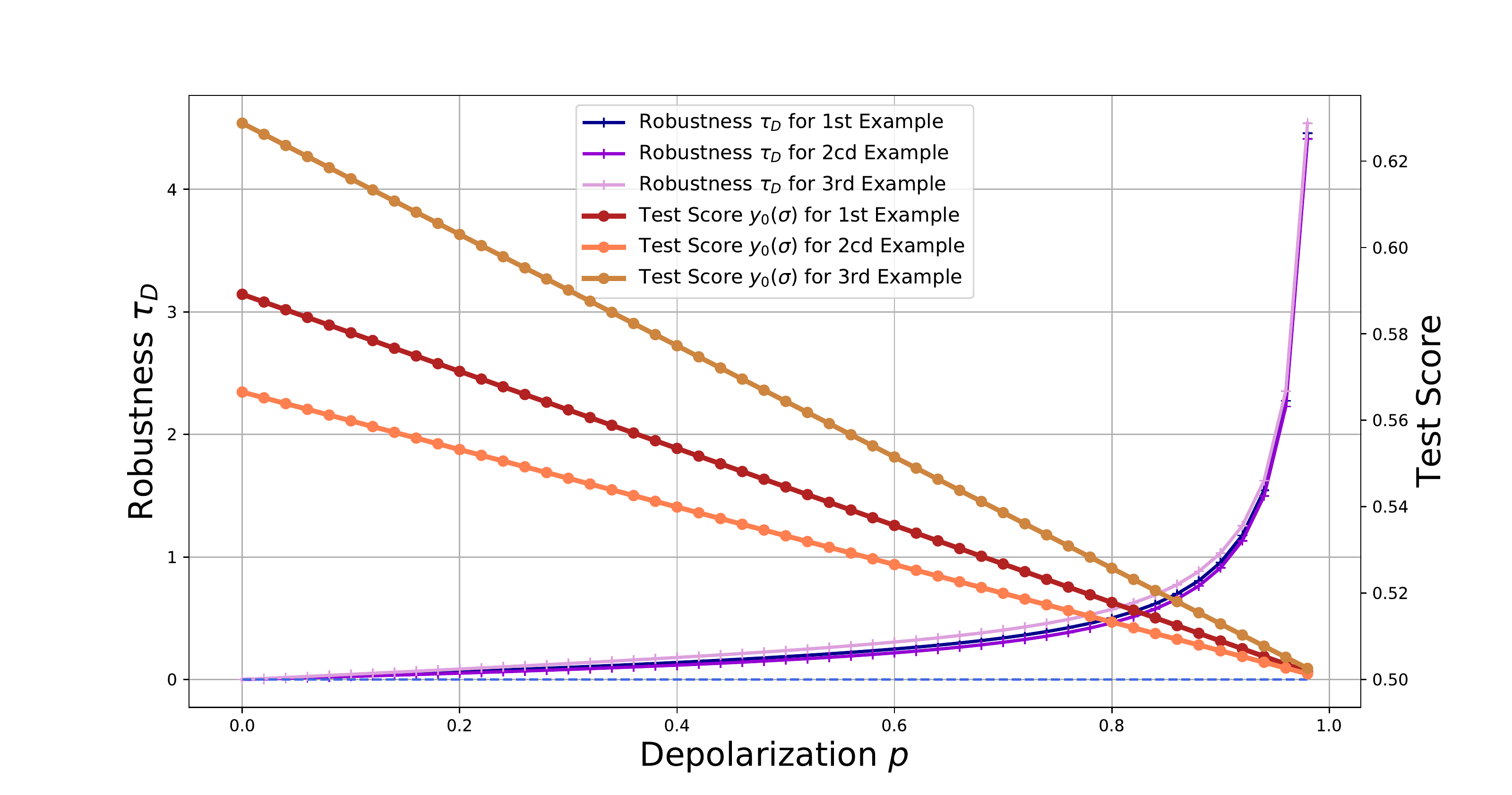}
\caption{\small{For three different test examples, we see how the simulated results for robustness $\tau_D$ and test scores $\tilde{\vec y}_0(\sigma)$} varies with respect to $p$. All three test examples have label $0$, where the test score is above the blue dotted line. The closed-form expressions for the variation of $\tau_D$ and $\tilde{\vec y}_0(\sigma)$ are shown in Eqs.~\eqref{eq:taudgraph} and ~\eqref{eq:ykgraph}.} 
\label{fig:p_tau_test} 
\end{figure}

\subsection{Evaluation metrics and adversarial attack methods}
To evaluate the performance of our protocol, we adopt an adversarial attack method that is widely employed in classical machine learning. It is known as the iterative-fast gradient sign method (I-FGSM) with $l_2$-bounded norm \cite{liu2016delving,dong2018boosting,madry2018towards} that aims to attack the test dataset $\mathcal{D}_{\text{Te}}$ to make incorrect predictions when using a trained classifier. If we denote the original input by $\vec x$ and the adversarial example at the $t^{\text{th}}$ updating step when using the I-FGSM by $\vec x'(t)$, then 
\begin{align}\label{eqn:adv}
	& \vec{x'}_{(0)} = \vec x \nonumber \\
&\vec {x'}_{(t+1)} = \vec {x'}{(t)} + \alpha\cdot \text{sign}(\nabla_{\vec x}\mathcal{L}),
\end{align} 
where  $\alpha = L/T$ is the learning rate with $\|\vec x- \vec x'\|_2 \leq L$ and $\mathcal{L}$ is the loss function formulated in Eq.~\eqref{eqn:loss_opt}. 
\subsection{Adversarial attack at test time}
Here we employ our trained classifier and the adversarial attack method formulated above to quantify the performance of our protocol. Recall that Corollaries \ref{coro:DP_infi} and \ref{coro:DP_fini} are the special cases of Theorems \ref{thm:DP_infini} and \ref{thm:DP_fini} when applied to binary QNN classifiers and work in the regime of using infinite and finite sampling of the output probabilities respectively. Here we explore how our protocol protects the binary QNN classifier against adversarial attacks under these two settings.

\begin{table*}[]
\centering
\resizebox{0.75\textwidth}{!}{%
\begin{tabular}{c|c|c|c|c}
\hline
\hline
Infinite Precision Case  ($n_{samp}=\infty$) & $p=0, \tau_D=0 $ & $p^{(1)}, \tau_D^{(1)} $ & $p^{(2)}, \tau_D^{(1)}$ & $p^{(1)}, \tau_D^{(2)}$ \\ \hline
$\tilde{\vec y}_0(\rho)$ & $58.92\%$ (label $0$)   &   $52.96\%$   (label $0$) & $57.11\%$ (label $0$) &  $49.42\%$  (label $1$)\\ \hline
Finite Precision Case ($p^{(1)}, \tau_D^{(1)}$) & --- & $n_{samp}=50$ & $n_{samp}=500$ & $n_{samp}=5000$ \\ \hline
$\tilde{\vec y}_0(\rho)$ & --- & 44.32\% (label 1) & 55.80\% (label 0) & 53.88\% (label 0) \\ \hline
\hline
\end{tabular}%
} 
\caption{\small{We list the test scores of the selected test example $\{\rho, \tilde{y}\}$ after bounded-norm adversarial attacks in both the infinite and finite sampling cases. The parameters $p=0, \tau_D=0$ refers to the test score under in the absence of any depolarisation noise in the circuit. The other parameter settings are  $\{p^{(1)}=0.5, \tau_D^{(1)}=0.02 \} $, $\{p^{(2)}=0.1, \tau_D^{(1)}=0.02 \} $ and $\{p^{(1)}=0.5, \tau_D^{(2)}=0.2\}$.}} 
\label{tab:my-table}
\end{table*}  

\textit{The infinite sampling case.}  At testing time, we randomly sample an example  $(\rho=\ket{\vec x}\bra{\vec x}, \tilde{y})$  from $\mathcal{D}_{Te}$ to investigate its robustness $\tau_D$ with respect to different level of depolarisation noise $p$.  Without loss of generality,  the original test example has label $\tilde{y}=0$.  We set three different values of $p$ and $\tau_D$: $\{p^{(1)}=0.5, \tau_D^{(1)}=0.02 \} $; $\{p^{(2)}=0.1, \tau_D^{(1)}=0.02 \} $ and $\{p^{(1)}=0.5, \tau_D^{(2)}=0.2\}$.  From Eq.~(\ref{eq:dpdefdep}), their corresponding privacy budgets are $\epsilon_1=1.04$, $\epsilon_2=1.36$ and $\epsilon_3=1.4$. Given our input $\rho$, the outputs of our trained classifier with added depolarizsation noise are $\Pr(\tilde{\vec y}^{(1)}(\rho)=0)=54.46\%$, $\Pr(\tilde{\vec y}^{(2)}(\rho)=0)=58.04\%$ and  $\Pr(\tilde{\vec y}^{(3)}(\rho)=0)=54.46\%$, where the corresponding constants $B$ defined in Corollary \ref{coro:DP_infi} is $B^{(1)}=1.20$, $B^{(2)}=1.38$, and $B^{(3)}=1.20$, respectively.  Following the condition for robustness in Eq.~\eqref{eqn:B_infi}, we have confidence that the classifier is robust to adversarial attacks if $B> e^{2\epsilon}$. A simple comparison indicates that robustness is guaranteed when $\{p=0.5, \tau_D=0.02 \} $, since $B^{(1)} >e^{2\epsilon_1}=1.08$ while $B^{(2)}<e^{2\epsilon_2}=1.85$ and $B^{(3)}<e^{2\epsilon_3}=1.96$. 

To validate the correctness  of our theoretical results, we employ I-FGSM to attack our trained classifier, where we identify the $l_2$-norm bound with its corresponding $\tau_D$ value. The left panel of  Figure~\ref{fig:attck_inifi} demonstrates the simulation results and Table~\ref{tab:my-table} shows the final test score of the attacked input.  The classifier with the first  setting $\{p^{(1)}=0.5, \tau_D^{(1)}=0.02 \} $  is robust to the bounded-norm adversarial attacks,  where the predicted label of $\tilde{\bm{x}}$ is still `0'.  For the third setting when $\{p^{(3)}=0.5, \tau_D^{(3)}=0.2\}$,  the adversary can easily perturb the input and lead the classifier to give the wrong prediction.  In particular, the adversary can easily perturb the input to cross the classification boundary, as highlighted by the purple line.  For the second setting with $\{p^{(2)}=0.1, \tau_D^{(2)}=0.05 \}$, the classifier correctly predicts the label, while our protocol cannot provide any promises, since Theorem \ref{thm:DP_infini} and Corollary \ref{coro:DP_infi} provides only sufficient conditions for robustness.  The above three simulation results are then in accordance with our theoretical results.   
\begin{figure*}
\includegraphics[width=0.48\textwidth]{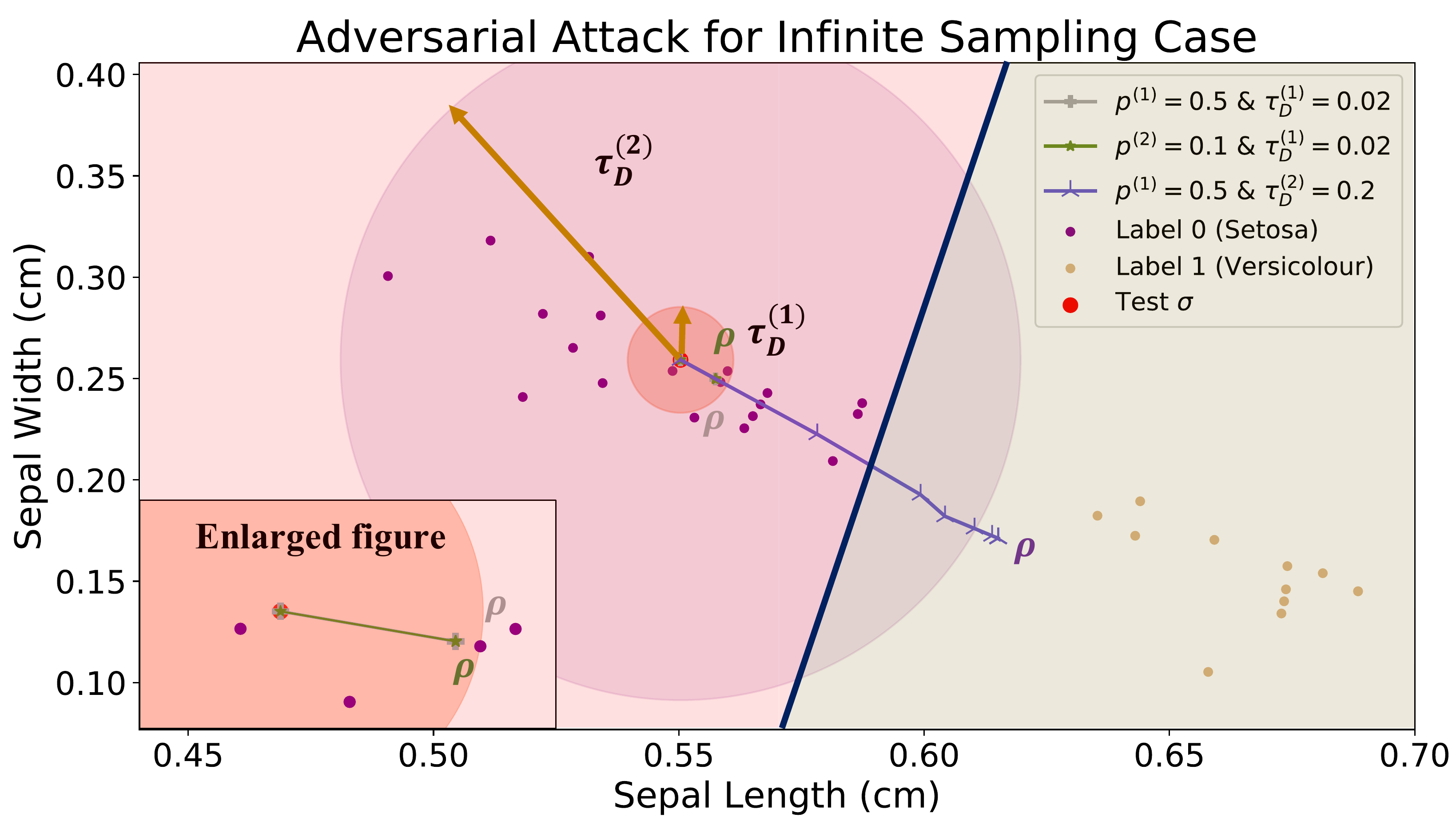}
\includegraphics[width=0.49\textwidth]{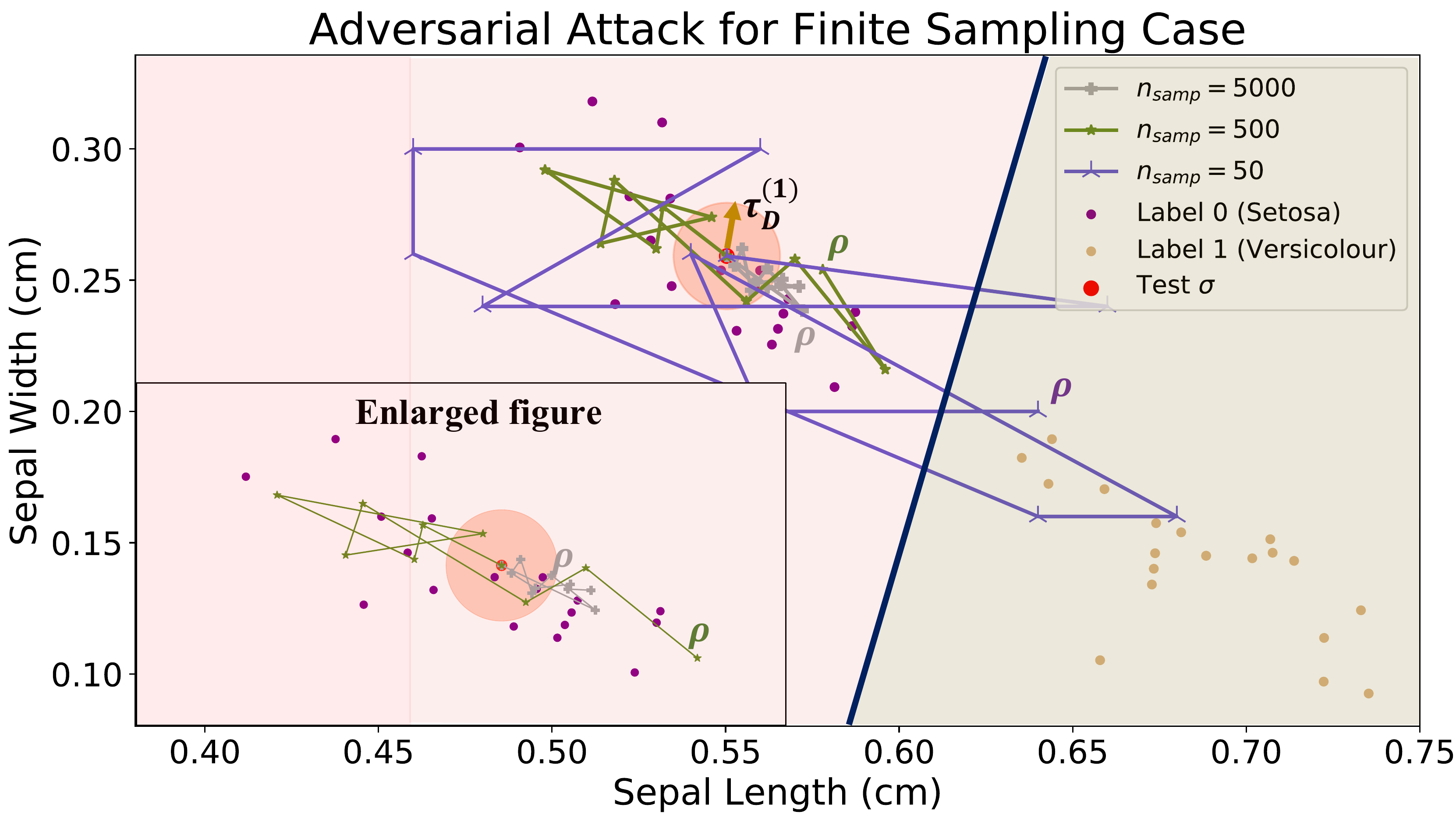} 
\caption{\small{\textit{The robustness of our protocol to adversarial examples.} (\textit{Left}) The left panel illustrates a bounded-norm attack on the Iris dataset in the infinite sampling case. The inner and outer circle regions indicate the robustness values $\tau_D^{(1)}=0.02$  and $\tau_D^{(2)}=0.2$ respectively. The thick purple line is a trained hyperplane of our QNN classifier. The dotted arrows indicate how an adversary iteratively attacks the input $\rho$ under three different settings of  $\{p^{(1)}=0.5, \tau_D^{(1)}=0.02 \} $; $\{p^{(2)}=0.1, \tau_D^{(1)}=0.02 \} $ and $\{p^{(1)}=0.5, \tau_D^{(2)}=0.2\}$, where the aim of the adversary is  to induce the classifier to output the wrong prediction.  The inner plot enlarges the part of the central figure near the test example. (\textit{Right}) The right panel illustrates the bounded-norm attack in the finite precision case. The circle region indicates the robustness value $\tau_D^{(1)}=0.02$ and $p^{(1)}=0.5$. The dotted arrows indicate the path of an adversary that iteratively attacks the input under $n_{samp}=50, 500$ and $5000$, where the adversary aims to induce  the classifier to output the wrong prediction. The inner plot enlarges the part of the central figure near the test example.}} 
\label{fig:attck_inifi} 
\end{figure*}

\textit{The finite sampling case.}   The only difference in the finite sampling case is the acquisition of the output of our trained classifier. The same test example $(\rho=\ket{\vec x}\bra{\vec x}, \tilde{y})$ is employed. The hyperparameters are set as  $\{p^{(1)}=0.5,  \tau_D^{(1)}=0.02 \}$ and from Eq.~\eqref{eq:dpdefdep}, the privacy budget $\epsilon=1.04$ is fixed. We set three different sampling number values $n_{samp}$ to explore how $n_{samp}$ affects the robustness guarantees, where $n_{samp}^{(1)}=50$, $n_{samp}^{(2)}=500$, and $n_{samp}^{(3)}=5000$. The corresponding three approximated test scores are  $\Pr(\tilde{\vec{y}}^{(1)}=0)=0.515$, $\Pr(\tilde{\vec{y}}^{(2)}=0)=0.529$ and $\Pr(\tilde{\vec{y}}^{(3)}=0)=0.552$.   The corresponding parameters $B$ are $B^{(1)}=1.06$, $B^{(2)}=1.124$, and $B^{(3)}=1.23$ with respect to $n_{samp}^{(1)}$,  $n_{samp}^{(2)}$, and $n_{samp}^{(3)}$.  Following the results of Theorem \ref{thm:DP_fini} and Corollary \ref{coro:DP_fini}, with probability at least $1-2\exp\left(-{2n_{samp}\zeta^2} \right)$, the trained classifier with added depolarisation noise is robust to adversarial attacks if ${B}> e^{2\epsilon}$. By setting $\zeta=0.95$, a simple inspection shows that  $n_{samp}=5000 $ guarantees robustness. Analogous to the infinite sampling case, we employ a bounded-norm adversary to confirm the correctness of our theory result, where the simulation results are shown in the right panel of Figure~\ref{fig:attck_inifi}. 

For more details on the implementation of the classifier and perfomance analysis of our protocol please see Appendix~\ref{app:h}.  
\section{Advantages of protocol}
Adversarial settings naturally occur when data needs to be delegated to different parties, for instance in a client-server setting and in multiparty computing. When this data is in the form of quantum states before processing using a quantum classifier, our protocol currently provides the only exisiting method to protect the general quantum classifier against arbitrary adversarial examples and also includes a theoretically provable bound. Furthermore, it can take advantage of certain exisiting quantum noise in a quantum classifier, like depolarisation noise, to provide protection against adversarial examples thus obviating the need for error-correction or error-mitigation if no other noise sources are present. Moreover, even if the test score is diminished in presence of depolarisation noise, its original value in the absence of any quantum noise can be retrieved by simply increasing the number of times one samples from the classifier. This sample complexity increases with the amount of exisiting depolarisation noise and is independent of the dimension of the state itself. 

Utilizing quantum noise like depolarisation noise also has certain advantages over classical methods for classical data in improving robustness against adversarial examples. We discuss this below. 
\subsection{Comparison to the best known classical protocol} 
While in the quantum case the theoretical bound on robustness is independent of the details of the classification model and is simple to compute, this is not true in the best known classical protocol. Before elaborating on this quantum advantage, we briefly review the classical results. 

Following the results of \cite{lecuyer2019certified}, classical $\epsilon$-differential privacy of a classification algorithm is obtained by adding noise sampled from the Laplacian distribution $\mathcal{N}(\bm{z},\kappa)$ to the trained classifier. This is commonly known as the Laplace mechanism. For numerical functions \footnote{For non-numerical functions, the exponential mechanism is employed.}, the only other common method to attain differential privacy is the Gaussian mechanism, which adds noise sampled from the Gaussian distribution. However, this leads to classical $(\epsilon, \delta)$-differential privacy where $\delta \neq 0$, so cannot be directly compared to our quantum scenario where $\delta=0$. The Laplacian distribution used in the Laplace mechanism can be written as
 \begin{align}\label{eqn:lap_mec}
  \mathcal{N}(\bm{z},\kappa)= \frac{\sqrt{2}}{2\kappa}\exp{\left(\frac{-|\bm{z}|}{\sqrt{2}\kappa}\right)}~, \text{with}~\kappa= \frac{\Delta f L}{\epsilon},
 \end{align}
where $\kappa $ refers to the variance of the Laplacian distribution and $L$ is the upper-bounded $l_2$ norm between original input $\vec x$ and attacked input $\vec x'$ such that classical $\epsilon$-differential privacy is preserved. The sensitivity $\Delta f$ of the function $f(\cdot)$ applied at a layer of the neural network classifier just before the Laplacian noise is injected is defined as 
\begin{align} \label{eq:deltaf}
\Delta f=\max_{\vec x, \vec x'}||f(\vec x)-f(\vec x')||_2/||\vec x-\vec x'||_2.  
\end{align}

The classical protocol runs in the following way. In the testing phase, the adversarial example $\vec x'$, where $\|\vec x-\vec x'\|_2\leq L$ and $\vec x$ is the original test example, is inserted into the trained classifier $\vec y(\cdot)$. The predicted label for $\vec x'$ is obtained by invoking $\vec y(\vec x')$ a total $N$ times. For every run of $\vec y(\vec x')$, the noise $z_{i,j}$ with $i=1,...,N$ is independently sampled from $\mathcal{N}_L(\bm{z},\kappa)$ and applied to the input to some layer $j$ of the neural network realising the classifier. Let $N_k$ denote the number of times that the predicted label is $k$, so the probability of the predicted label being $k$ is given by $N_k/N$. Then, similarly to Theorem \ref{thm:DP_fini}, we can write the following condition for robustness of the $K$-class classifier under the Laplace mechanism. 

\begin{lemma}[modified from \cite{lecuyer2019certified}]\label{lem:lap}
	Let $\vec x$ be the input to the $K$-multiclass classifier, which is endowed with classical $\epsilon$-differential privacy under the Laplace mechanism, with $\epsilon=\Delta f L/\kappa$, as formulated in Eq.~\eqref{eqn:lap_mec}. Let $C$ be the label of $\vec x$. Then with probability at least $1-\zeta$, the classifier is robust to any adversarial example $\vec x'$ with $\|\vec x-\vec x'\|_2\leq L$ if 
\begin{equation}\label{eqn:lem_lap_1}
	L=\frac{\epsilon \kappa}{\Delta f}< \frac{\kappa}{2 \Delta f}\ln \left(\frac{\frac{N_C}{N} -\sqrt{\frac{1}{2N}\ln\left(\frac{2}{1-\zeta} \right)}}{\max_{k \neq C} \frac{N_k}{N} + \sqrt{\frac{1}{2N}\ln\left(\frac{2}{1-\zeta} \right)}}\right).
\end{equation} 
\end{lemma}  
This means that this best available classical theoretical bound to $L$ depends on $\Delta f$, which in general is dependent on both the details of the classification model used and the layer of the neural network in which the Laplacian noise is injected. However, in the quantum scenario with depolarisation noise, we see that the robustness bound is independent of both $U$, the circuit realising the quantum classifier, as well as the location or locations of noise injection. This means that the adversarial robustness bound is universal for all quantum classifiers. 

We can see this from the fact that the final state of the quantum circuit after applying depolarisation noise in layers $1$ to $l$ depends only on the product $\prod_{i=1}^l(1-p_i)$, which is independent of $U$ and invariant under any re-ordering of the layers. This simplicity in the quantum case results from two facts: that the `noisy' part of depolarisation noise lies in injecting a maximally-mixed channel with a certain probability and that unitary $U$ operations realising any quantum classifier are unital (i.e., the identity operator $\mathbf{1}$ remains invariant under $U$). On the other hand, there is no known classical equivalent of this property that also gives rise to differential privacy. 

The dependence of $\Delta f$ on the details of the classifier in the most general cases also leads to a difficulty in the computation of $\Delta f$ and is often intractable except in the simplest cases \cite{lecuyer2019certified}. This means that, unlike in the quantum case, the corresponding classical bound on robustness $L$ cannot be derived in closed form from Eq.~\eqref{eqn:lem_lap_1} in the most general case. 

However, in special simple cases we can provide quantitative examples of this quantum advantage. As a simple illustration, we can look at the binary classifier for the kernel perceptron, which can be written as 
\begin{align}
\vec y(\vec x)=
\begin{pmatrix}
\vec y_0(\vec x)\\
1-\vec y_0(\vec x)
\end{pmatrix}, ~ \vec y_0(\vec x) = \sum_{i=1}^M w_i^*y_i^*K(\vec x^*_i, \vec x),
\end{align}
where $\{(\vec x_i^*, y^*_i)\}_{i=1}^M$ are the $M$ training examples and $\{w_i^*\}_{i=1}^M $ are trained parameters of the classifier. We can consider the polynomial kernel
\begin{align}
K(\vec x^*_i, \vec x)=(\vec x^*_i \cdot \vec x)^n,
\end{align}
where $n$ is the kernel degree and $n=1$ is the special case of the linear kernel. We now have the following theorem.

\begin{theorem}\label{thm:thm3}
	We have a binary classifier $\vec y(\vec x)=(\vec y_0(\vec x), 
1-\vec y_0(\vec x))^T$ where $\vec y_0(\vec x)=\sum_{i=1}^M w^*_i y_i K(\vec x^*_i, \vec x)$ with the polynomial kernel $K(\vec x^*_i, \vec x)=(\vec x^*_i \cdot \vec x)^n$. Let $\vec x$ denote all correctly labelled test examples. We now implement the Laplace mechanism in this classifier where the sensitivity is $\Delta f \equiv ||\vec y(\vec x)-\vec y(\vec x')||_2/||\vec x-\vec x'||_2$ and the privacy budget is $\epsilon=\Delta f L/\kappa$. Let us choose $\tilde{\vec y}_0(\vec x)>\exp(2\epsilon)\tilde{\vec y}_1(\vec x)$ and define $B \equiv \vec y_0(\vec x)/\vec y_1(\vec x)$. We can define the function $g(\cdot)$ for our noisy classifier where $g(B)=\tilde{\vec y}_0(\vec x)/\tilde{\vec y}_1(\vec x)$. Then the classifier is robust under any adversarial example $\vec x'$ where $||\vec x-\vec x'||_2\leq L$ and 
\begin{align} \label{eq:L2classicalbound}
L\leq \frac{1}{M}\frac{\kappa}{2\sqrt{2}n\max\{|w^*_i y_i|\}_{i=1}^M}\ln g(B).
\end{align}
\end{theorem} 
\begin{proof}[Proof of Theorem \ref{thm:thm3}]
	We compute an upper-bound for $\Delta f$ in terms of classification model parameters in $\vec y(\vec x)$ and use $L=\epsilon \kappa/\Delta f<\kappa \ln g(B)/(2\Delta f)$. Please see Appendix~\ref{sec:lemma2proof} for details. 
\end{proof} 

From this we see that we can guarantee only a smaller robustness bound for a more nonlinear kernel (i.e., higher $n$). We can also use a quantum classifier below to realise the same polynomial kernel and find a robustness bound that is now independent of degree of nonlinearity of the kernel. 

\begin{theorem}\label{thm:thm4}
We have a kernel perceptron binary classifier $\vec y(\sigma)=(\vec{y}_0(\sigma),
\vec{y}_1(\sigma)=1-\vec{y}_0(\sigma))^T$ that is realised by a quantum circuit in the absence of noise and takes the form in Figure~\ref{fig:noise1} with $D_{meas}=2$. Without losing generality, we can assume the class label of $\sigma$ is $0$. Now we add depolarisation noise channels $\mathcal{N}_{p_i}$ to the classifier where $i=1,...,l$ to create a noisy classifier $\tilde{\vec y}(\sigma)$. Let us choose $\tilde{\vec y}_0(\sigma)>\exp(2\epsilon)\tilde{\vec y}_1(\sigma)$ and define $B\equiv \vec{y}_0(\sigma)/\vec{y}_1(\sigma)$. Then the noisy classifier is robust under any adversarial perturbation $\sigma \rightarrow \rho$ such that $\tau(\sigma, \rho)\leq \tau_D$ where
\begin{align}
\tau_D< \frac{B-1}{4(B+1)+8(1-p)/p}
\end{align}
for $p=1-p \in (0, 1/2)$ and 
\begin{align}
\tau^2_D <\frac{B-1}{4(B+1)(1-p)/p+8(1-p)^2/p^2}
\end{align}
for $p=1-p \in [1/2, 1)$. 
\end{theorem}

\begin{proof}[Proof of Theorem \ref{thm:thm4}]
	We use the expression for $\epsilon$-quantum differential privacy with depolarisation noise that relates $\tau_D$ with $\epsilon$ and relate $\epsilon$ to the fraction $B$. Please see Appendix~\ref{sec:lemma3proof} for details. 
\end{proof}	
The trace distance $\tau_D$ can be turned into a corresponding $l_2$ norm distance $L$ if an encoding of the classical data $\vec x$ into a quantum state $\sigma_{\vec x}$ is chosen. For instance, we can choose the most widely used amplitude encoding $\vec x \rightarrow \sum_{i=1}^D x_i \ket{i}$ where $x_i$ is the $i^{\text{th}}$ element of $\vec x$ and we assume for simplicity the normalisation $||\vec x||_2=1$. Then the trace distance $\tau(\sigma_{\vec x}, \sigma_{\vec x'})=\sqrt{1-\text{Tr}(\sigma_{\vec x}\sigma_{\vec x'})}=\sqrt{1-(\vec x \cdot \vec x')^2}$ and $l_2 \equiv ||\vec x-\vec x'||_2=\sqrt{2-2(\vec x \cdot \vec x')}$. Therefore we can write $\tau(\sigma_{\vec x}, \sigma_{\vec x'})=l_2\sqrt{1-l_2^2/2}\geq l_2$. This means Theorem \ref{thm:thm4} still holds if we replace $\tau_D$ with $L$ and can compare results directly with Theorem \ref{thm:thm3} with the same chosen constant $B$. Then we see how the robustness bound in the classical case is dependent on details of the kernel function like the nonlinearity $n$ whereas the robustness bound can be completely independent of the kernel function.

While in Theorems \ref{thm:thm3} and \ref{thm:thm4} we have provided only sufficient though not necessary conditions for robustness, this was only for the purpose of illustrating a clearer interpretation of robustness in terms of a model parameter like the degree of nonlinearly $n$. Necessary conditions can also be found since we already have the exact expressions for $L$ and $\tau_D$. We know that the former is dependent on the details of the classification model through $\Delta f$ in the most general case whereas the latter is dependent only on $p$, $D_{meas}$ and $\epsilon$, which can be chosen to be constants independent of the details of the kernel or any other classifier. This latter property of $\tau_D$ we have already learned is not consistent with any known classical mechanism for differential privacy. 

Another advantage of the quantum mechanism is that depolarisation noise can occur naturally in quantum systems especially for NISQ devices, whereas the Laplace mechanism needs to be artifically injected into the classifier. From \cite{zhou2017differential} it is known that other quantum noise like amplitude and phase damping and generalised amplitude damping also have the $\epsilon$-quantum differential privacy property in the single-qubit case. It remains exciting work for future investigation to see in the general multiqubit case if other natural sources of quantum noise can be harnessed for adversarial protection. 

\section{Discussion}
We demonstrated how depolarisation noise placed anywhere in a quantum circuit used for classification can be exploited to protect the classification algorithm against arbitrary `worst-case' attacks like adversarial examples. A theoretical bound for robustness can be proved without any assumptions on the type of adversary or the classification model and applies to both quantum and classical data. This bound relies on a new relationship we introduced between quantum differential privacy and adversarial robustness in the quantum setting. In particular, depolarisation noise allows the theoretical robustness bound to be dependent only on the number of classes in the classification model and no other feature of the classifier. However, all known classical noise that can give rise to differential privacy results in robustness bounds that would generally depend on more details of the classification model, for instance the degree of nonlinearity of the classification boundary. 

This result raises many intriguing possibilities for exploring other naturally-occuring quantum noise sources that could offer similar advantages against adversarial attacks, which become pertinent concerns as quantum data are shared in a future quantum internet. We see that the fruitful merging of concepts in security and quantum machine learning potentially leads to quantum advantages that is independent of quantum speedups. This also highlights how noise in the NISQ era for quantum computation can be used as a positive feature and can be employed in parallel with other methods to demonstrate quantum advantage. 
\section*{Acknowledgements}
NL is grateful to Barry Sanders (University of Calgary), Ryan LaRose (Michigan State University), Dong-Ling Deng (Tsinghua University), Jeongho Bang (Korea Institute for Advanced Study) and Srinivasan Arunachalam (IBM Research) for very interesting and fruitful discussions. NL also thanks Dong-Ling Deng for giving useful feedback on this manuscript during her stay in Tsinghua University. 
\bibliographystyle{plain}
\bibliography{myref2}
\newpage 
\appendix
\section{Proof of Lemma 1} \label{sec:lemma1}
Here we prove that if the noiseless quantum classifier assigns $\sigma$ to the class $C$, i.e., $C=\arg \max_k \vec{y}_k(\sigma)$, then the noisy circuit with depolarisation noise also assigns $\sigma$ to the class $C$, i.e. $C=\arg \max_k \tilde{\vec{y}}_k(\sigma)$. This is equivalent to the condition that if $\vec{y}_C(\sigma)>\max_{k \neq C}\vec{y}_k(\sigma)$, then  $\tilde{\vec{y}}_C(\sigma)>\max_{k \neq C}\tilde{\vec{y}}_k(\sigma)$. 

Using Eq.~\eqref{eq:noisedep}
\begin{align}
\mathcal{N}_{p}(\sigma)=\frac{p}{D}\mathbb{I}_{D}+(1-p)\sigma, 
\end{align}
we can rewrite Eq.~\eqref{eq:tildeyk1} as 
\begin{align}\label{eq:tildeyk2}
&\vec{\tilde{y}}_k(\sigma) \equiv \text{Tr}(\Pi_k \mathcal{N}_{p_l}(U_L(...\mathcal{N}_{p_1}(U_1(\sigma \otimes |a\rangle \langle a|)U_1^{\dagger})...))) \nonumber \\
&= \text{Tr}(\Pi_k p\frac{\mathbb{I}}{D})+(1-p)\vec{y}_k(\sigma) \nonumber \\
&=\frac{p}{K}+(1-p)\vec{y}_k(\sigma),
\end{align}
where $k=1, 2, ...,K$ and $p \equiv 1-\prod_{i=1}^l (1-p_i)$. The second line can be readily derived by induction. Then if $\vec{y}_C(\sigma)>\max_{k \neq C}\vec{y}_k(\sigma)$, Eq.~\eqref{eq:tildeyk2} implies 
\begin{align}\label{eq:Append_1_3}
& \tilde{\vec{y}}_C(\sigma)=\frac{p}{K}+(1-p)\vec{y}_C(\sigma) \nonumber \\
&>\frac{p}{K}+(1-p)\max_{k \neq C}\vec{y}_k(\sigma) \nonumber \\
&=\max_{k \neq C}\left(\frac{p}{K}+(1-p)\vec{y}_k(\sigma)\right)=\max_{k\neq C}\tilde{\vec{y}}_k(\sigma).
\end{align} 
Conversely, if $\tilde{\vec{y}}_C(\sigma)>\max_{k \neq C}\tilde{\vec{y}}_k(\sigma)$, then from Eq.~\eqref{eq:tildeyk2} it is clear that $\vec{y}_C(\sigma)>\max_{k \neq C}\vec{y}_k(\sigma)$ is true also. \qed 

\section{Proof of Proposition 1} \label{sec:proposition1}
From Lemma \ref{lem:K-multi-classification}, we know that if $\sigma$ is labelled as $C$ in the noiseless circuit then in the infinite sampling limit this label is maintained in the corresponding circuit with depolarisation noise, so $\tilde{\vec{y}}_C(\sigma)>\max_{k \neq C}\tilde{\vec{y}}_k(\sigma)$. However, in the finite sampling limit with sample complexity $N$, we only have access to the estimate $\tilde{\vec{y}}_k^{(N)}(\sigma)$. So we want to find the smallest $N$ so $\tilde{\vec{y}}_C^{(N)}(\sigma)>\max_{k \neq C} \tilde{\vec{y}}_k^{(N)}(\sigma)$ with probability at least $\beta$. 

From results in Lemma \ref{lem:K-multi-classification}, we see that since $\tilde{\vec{y}}_k(\sigma)=1-p/K+p\vec y_k(\sigma)$ and $\xi \equiv \vec{y}_C(\sigma)-\max_{k \neq C} \vec{y}_k(\sigma)$, then $\eta \equiv \tilde{\vec{y}}_C-\max_{k \neq C} \tilde{\vec{y}}_k(\sigma)=p\xi$. Thus we need large enough sampling to resolve the difference $\tilde{\vec{y}}_C^{(N)}(\sigma)-\max_{k \neq C} \tilde{\vec{y}}_k^{(N)}(\sigma)$ to at least $2\eta=2p\xi$. It is then sufficient to find $N$ that estimates $\tilde{\vec{y}}_k^{(N)}(\sigma)$ to precision $2\eta$. To find $N$, we can employ Hoeffding's inequality in the following. \\
\noindent \textbf{Lemma A.} (Hoeffding's  inequality \cite{mohri2018foundations})\label{lem:hoef}
Let $Z_1,...,Z_N$ be independent bounded random variables	with $Z_i\in[a,b]$ for all $i\in[N]$, where  $-\infty< a \leq b < \infty$. Then the probability 
\begin{align} \label{eq:hoeffding}
	\Pr\left(\left| \frac{1}{N}\sum_{i=1}^NZ_i-\mathbb{E}(Z_i) \right| \leq \zeta \right)\geq 1-2\exp\left(-\frac{2N\zeta^2}{(b-a)^2} \right).
\end{align}
In our case, we can use $b-a=1$, $\zeta=2\eta$, $(1/N)\sum_{i=1}^N Z_i=\tilde{\vec{y}}_k^{(N)}(\sigma)$ and $\mathbb{E}(Z_i)=\tilde{\vec{y}}_k(\sigma)$. Thus if we require the probability $\Pr(|\tilde{\vec{y}}_k^{(N)}(\sigma)-\vec{y}_k(\sigma)|<2\eta)\geq \beta$, it is sufficient to require $1-2\exp(-8N(1-p)^2\xi^2)\sim \beta$, or eequivalently
\begin{align}
N \sim \frac{1}{8(1-p)^2\xi^2}\ln \left(\frac{2}{1-\beta}\right). 
\end{align}
\qed 
\section{Proof of Lemma 2} \label{sec:lemma2proof} 
This proof follows Zhou and Ying \cite{zhou2017differential}, applied to the case where the dimension of the final projector is $D_{meas}$ and we can apply multiple depolarisation channels $\mathcal{N}_{p_i}$ for $i=1,...,l$ where $p \equiv 1-\prod_{i=1}^l (1-p_i)$. To show $\epsilon$-differential privacy, we must show that when $\tau(\sigma,\rho)\leq \tau_D$, the following relation must hold, i.e., 
\begin{align}
e^{-\epsilon} \leq \frac{\tilde{\vec y}_k(\rho)}{\tilde{\vec y}_k(\sigma)}\leq e^{\epsilon},
\end{align}
where from Eq.~\eqref{eq:tildeyk1}
\begin{align}\label{eqn:lem1_1}
	\tilde{\vec y}_k(\rho)=\Tr( \Pi_k (\mathcal{N}_{p_l}(U_l(...\mathcal{N}_{p_1}(U_1(\rho)U_1^{\dagger})...U_l^{\dagger})))).
\end{align}
By employing the definition of depolarisation noise with noise parameter $p$ acting on an arbitrary quantum state $\sigma$, from Eq.~\eqref{eq:noisedep}    
\begin{align}
\mathcal{N}_{p}(\sigma)=\frac{p}{D}\mathbb{I}_{D}+(1-p)\sigma, 
\end{align}
we can derive
\begin{align}\label{eqn:lem1_3}
& \tilde{\vec y}_k(\rho)=\frac{p(D-D_{meas})}{D}\Tr(\Pi_k) \nonumber \\
&+(1-)p\text{Tr}(\mathbb{I}_{D-D_{meas}}\otimes \Pi_k U(\rho)U^{\dagger}), 
\end{align}
and similarly for $\tilde{\vec y}_k(\sigma)$. From this we can write 
\begin{align}
&\frac{\tilde{\vec y}_k(\rho)}{\tilde{\vec y}_k(\sigma)}-1= \nonumber \\
&(1-p)\text{Tr}(U(\rho-\sigma)U^{\dagger})\mathbb{I}_{D-D_{meas}}\otimes \Pi_k)/\nonumber\\
&\left(\frac{p(D-D_{meas})}{D}\text{Tr}(\Pi_k)+F\right) \nonumber \\
&\nonumber \\
\leq &\frac{(1-p)\tau_D \text{Tr}(\mathbb{I}_{D-D_{meas}}\otimes \Pi_k)}{\frac{p(D-D_{meas})}{D}\text{Tr}(\Pi_k)} \nonumber\\
 = & \frac{1-p}{p}D_{meas}\tau_D~,
\end{align}
where $F \equiv (1-p)\text{Tr}(\mathbb{I}_{D-D_{meas}}\otimes \Pi_k U(\sigma)U^{\dagger})>0$. In the first inequality we used the relation $\text{Tr}(U(\rho-\sigma)U^{\dagger}\Lambda_k)\leq \tau_D \text{Tr}(\Lambda_k)$ and the inequality $\tau(U(\sigma)U^{\dagger}, U(\rho)U^{\dagger})\leq \tau(\sigma,\rho)\leq \tau_D$ \cite{zhou2017differential, nielsen2010quantum}.    

To satisfy Eq.~(\ref{eqn:lem1_1}), we upper-bound this final term by $e^{\epsilon}-1$ and find the privacy budget 
\begin{align}
	& \frac{1-p}{p}D_{meas}\tau_D \leq e^{\epsilon}-1 \nonumber \\
	\Rightarrow & \epsilon = \ln\left(1+D_{meas}\tau_D\frac{1-p}{p}\right). 
\end{align} 
\section{Proof of Theorem 1} \label{sec:appendixc}
Here we prove that if $\tilde{\vec{y}}_k(\sigma)>e^{2\epsilon}\max_{k \neq C}\tilde{\vec{y}}_k(\sigma)$ where $\epsilon=\ln(1+D_{meas}(1-p)\tau_D/p)$, then $\tilde{\vec{y}}_C(\rho)>\max_{k \neq C}\tilde{\vec{y}}_k(\rho)$ for all $\rho$ where $\tau(\sigma,\rho)\leq \tau_D$. First we employ Lemma \ref{lem:DP-robust}, which states that given depolarisation noise with parameter $p$, the algorithm implemented by the noisy circuit has $\epsilon$-quantum differential privacy. Then from Eq.~\eqref{eq:def4example} following Definition \ref{def:QDP}, we see that in our case it states
\begin{align} \label{eq:ytilde}
e^{-\epsilon}\leq \frac{\tilde{\vec{y}}_k(\rho)}{\tilde{\vec{y}}_k(\sigma)}\leq e^{\epsilon},
\end{align} 
which holds true for when $\epsilon=\ln(1+D_{meas}(1-p)\tau_D/p)$ and all $\rho$ where $\tau(\sigma,\rho)\leq \tau_D$. Then if we insert $\tilde{\vec{y}}_k(\sigma)>e^{2\epsilon}\max_{k \neq C}\tilde{\vec{y}}_k(\sigma)$ into the above we can write 
\begin{align}
\tilde{\vec{y}}_k(\rho)\geq e^{-\epsilon}\tilde{\vec{y}}_k(\sigma)>e^{\epsilon}\max_{k \neq C}\tilde{\vec{y}}_k(\sigma).
\end{align}
Then the left-hand side inequality in Eq.~\eqref{eq:ytilde}, we find 
\begin{align}
\tilde{\vec{y}}_k(\rho)\geq \max_{k\neq C}\tilde{\vec{y}}_k(\rho).
\end{align}
From Lemma \ref{lem:K-multi-classification}, we see that this is also equivalent to the claim $\vec{y}_k(\rho)\geq \max_{k\neq C}\vec{y}_k(\rho)$. \qed
\section{Proof of Theorem 2} \label{sec:appendixd}
From Hoeffding's inequality (see Eq.~\eqref{eq:hoeffding} in Lemma A of Appendix~\ref{sec:proposition1}), it is clear that
\begin{align} \label{eq:yineq1}
\tilde{\vec{y}}_k^{(N)}(\sigma)-\zeta \leq \tilde{\vec{y}}_k(\sigma)\leq \tilde{\vec{y}}_k^{(N)}(\sigma)+\zeta
\end{align}
to probability greater than $1-2\exp(-2N\zeta^2)$. In the statement of Theorem \ref{thm:DP_fini}, we assume $\tilde{\vec{y}}_C^{(N)}(\sigma)-\zeta>e^{2\epsilon}\max_{k \neq C}(\tilde{\vec{y}}_k^{(N)}(\sigma)+\zeta)$. Inserting the above, this implies
\begin{align}
& \tilde{\vec{y}}_C(\sigma) \geq \tilde{\vec{y}}_C^{(N)}(\sigma)-\zeta \nonumber \\
&>e^{2\epsilon}\max_{k \neq C}(\tilde{\vec{y}}_k^{(N)}(\sigma)+\zeta)\geq e^{2\epsilon}\max_{k \neq C}\tilde{\vec{y}}_k(\sigma)
\end{align}
is true to probability at least $1-2\exp(-2N\zeta^2)$. From Theorem \ref{thm:DP_infini} we know that the above inequality $\tilde{\vec{y}}_C(\sigma)>e^{2\epsilon}\max_{k \neq C}\tilde{\vec{y}}_k(\sigma)$ leads to the condition $C=\arg \max_k \tilde{\vec{y}}_k(\rho)=\arg \max_k \vec{y}_k(\sigma)$ for all $\rho$ where $\tau(\sigma,\rho)\leq \tau_D$ and $\epsilon=\ln(1+D_{meas}(1-p)\tau_D/p)$. 

Then using Eq.~\eqref{eq:yineq1} again in the condition $C=\arg \max_k \tilde{\vec{y}}_k(\rho)$, equivalent to $ \tilde{\vec{y}}_C(\rho)>\max_{k \neq C}\tilde{\vec{y}}_k(\rho)$, we find 
\begin{align}
\tilde{\vec{y}}_C^{(N)}(\sigma)+\zeta>\max_{k \neq C}\tilde{\vec{y}}_k^{(N)}(\sigma)+\zeta
\end{align}
to probability at least $1-2\exp(-2N\zeta^2)$. 
\qed
\section{Proof of Theorem 3} \label{sec:theorem3proof}
We first observe that for integers $n>0$ and numbers $u$ and $v$ we have $u^n-v^n=(u-v)\sum_{j=0}^n u^j v^{n-1-j}$. If we assume $|u|, |v|\leq G$, this then implies $|u^n-v^n|\leq |u-v|nG^{n-1}$. Let $q(u_i)=a_i u_i^n$ so 
\begin{align} \label{eq:ineq1}
& |\sum_{i=1}^M q(u_i)-q(v_i)| \leq \sum_{i=1}^M |q(u_i)-q(v_i)| \nonumber \\
&\leq \sum_{i=1}^M|a_i||u_i^n-v_i^n|\leq \sum_{i=1}^M|a_i| |u_i-v_i|nG_i^{n-1},
\end{align}
where all $|u_i|, |v_i|\leq G_i$. In our case we can define $a_i=w^*_i y_i$, $u_i=\vec x^*_i \cdot \vec x$,  $v_i=\vec x^*_i \cdot \vec x'$. Suppose we fix a normalisation $||\vec x_i^*||_2=1=||\vec x||_2=|\vec x'||_2$. From the Cauchy-Schwarz inequality, $|u_i|=|\vec x^*_i \cdot \vec x|\leq ||\vec x_i^*||_2||\vec x||_2=1$ and similarly $|v_i|\leq 1$, so it is sufficient for us to choose $G_i=1$. We now want to compute the sensitivity for the kernel perceptron model where the sensitivity is defined in Eq.~\eqref{eq:deltaf}
\begin{align}
\Delta f=\max_{\vec x, \vec x'}||f(\vec x)-f(\vec x')||_2/||\vec x-\vec x'||_2,
\end{align}
where in our case of the polynomial kernel $f(\vec x)=\vec y(\vec x)=(\vec y_0(\vec x), 1-\vec y_0(\vec x))^T$ and $\vec y_0(\vec x)=\sum_{i=1}^M w^*_i y^*_i(\vec x^*_i \cdot \vec x)^n$. Then it is straightforward to show 
\begin{align}
\Delta f=\max_{\vec x, \vec x'} \sqrt{2}\frac{|\sum_{i=1}^M w^*_i y^*_i (K(\vec x^*_i, \vec x)-K(\vec x^*_i, \vec x'))|}{||\vec x-\vec x'||_2}.
\end{align}
Using Eq.~\eqref{eq:ineq1} for the polynomial kernel, we obtain
\begin{align} \label{eq:deltafineq}
& \Delta f \leq \sum_{i=1}^M \frac{|a_i| |u_i-v_i|}{||\vec x-\vec x'||_2}n=\sum_{i=1}^M \frac{|a_i| |\vec x^*_i\cdot (\vec x-\vec x'_i)|}{||\vec x-\vec x'||_2}n \nonumber \\
&\leq \sum_{i=1}^M\frac{|a_i| ||\vec x^*_i||_2||\vec x-\vec x'_i||_2}{||\vec x-\vec x'||_2}n \nonumber \\
&=\sum_{i=1}^M|w^*_i y_i|n\leq M \max\{|w^*_i y_i|\}_{i=1}^Mn,
\end{align}
where we used the normalisation $||\vec x^*_i||_2=1$ in the last line. In the special case of the linear kernel (or $n=1$) we have $\Delta f \leq M \max\{|w_y y_i|\}_{i=1}^M$. \\

From Eq.~\eqref{eqn:lap_mec} in the text, 
\begin{align}
\kappa=\frac{\Delta f L}{\epsilon}.
\end{align}
This means that the classifier is robust against all adversarial examples $\vec x'$ where $||\vec x'-\vec x||_2\leq L=\kappa \epsilon/\Delta f$. In our theorem, we require the condition $g(B) \equiv \tilde{\vec y}_0(\vec)/\tilde{\vec y}_1(\vec x) >\exp(2\epsilon)$ where $B \equiv \vec y_0(\vec x)/\vec y_1(\vec x)$, which gives $\epsilon<(1/2) \ln g(B)$. Together with $1/\Delta f \geq 1/(M \max\{|w^*_i y_i|\}_{i=1}^Mn)$ from Eq.~\eqref{eq:deltafineq}, this implies $||\vec x'-\vec x||_2 <\kappa \ln g(B)/(2\sqrt{2}M \max\{|w^*_i y_i|\}_{i=1}^Mn)$ is a sufficient condition for robustness.  
\section{Proof of Theorem 4} \label{sec:lemma3proof}
Following the results of Lemma \ref{lem:K-multi-classification}, when the depolarisation noise layer is inserted into the trained model just before the final measurement, the classifier $y(\sigma)$ has the $\epsilon$-differential privacy property where
\begin{align} \label{eq:DPconditionslast}
	&e^{-\epsilon}\tilde{\vec y}_0(\sigma) < \tilde{\vec y}_0(\rho)\leq e^{\epsilon}\tilde{\vec y}_0(\sigma) \nonumber \\
& e^{-\epsilon}\tilde{\vec y}_1(\sigma) \leq  \tilde{\vec y}_1(\rho)\leq e^{\epsilon}\tilde{\vec y}_1(\sigma).
\end{align}
Now, if the initial class label of $\sigma$ is 0, to correctly predict the attacked input $\rho$ requires
\begin{align}
	\tilde{\vec y}_0(\rho) > \tilde{\vec y}_1(\rho)=1-\tilde{\vec y}_0(\rho).   
\end{align} 
In combination with Eq.~\eqref{eq:DPconditionslast}, this robustness condition is equivalent to $(1+e^{2\epsilon})\tilde{\vec y}(\sigma) > e^{2\epsilon}$ or 
\begin{align} \label{eq:2epsiloncondition}
\tilde{\vec y}_0(\sigma)/\tilde{\vec y}_1(\sigma)>e^{2\epsilon}.
\end{align}
By including depolarisation channels with corresponding depolarisation parameters $p_1,...,p_l$, we can write 
\begin{align} \label{eqn:lemm1_3}
\tilde{\vec y}_0(\sigma)=p/2+(1-p)\vec y_0(\sigma),
\end{align}
where $p=1-p$. Then inserting Eq.~\eqref{eqn:lemm1_3} into Eq.~\eqref{eq:2epsiloncondition} we find
\begin{align}\label{eqn:lemm1_5}
 &  \frac{p}{2}+ (1-p)\vec y_0(\sigma) > e^{2\epsilon}\left(    \frac{p}{2} + (1-p)\vec y_1(\sigma)\right) \nonumber\\
   \Leftrightarrow & \frac{ \frac{p}{2}+ (1-p)\vec y_0(\sigma)}{ \frac{p}{2} + (1-p)\vec y_1(\sigma)} > e^{2\epsilon}\nonumber\\
   \Leftrightarrow & 1 + \frac{ (1-p)(\vec y_0(\sigma)-\vec y_1(\sigma))}{ \frac{p}{2} + (1-p)\vec y_1(\sigma)} >  \left(1+2\frac{1-p}{p}\tau_D \right)^2  \nonumber\\
   \Leftrightarrow &  \frac{ (1-p)(\vec y_0(\sigma)-\vec y_1(\sigma))}{ \frac{p}{2} + (1-p)\vec y_1(\sigma)} > 4\frac{1-p}{p}\tau_D +4\frac{(1-p)^2}{p^2}\tau_D^2   ~.
\end{align}
where we used $\epsilon=\ln(1+2\frac{1-p}{p}\tau_D)$ in the second line, which is a result from Lemma \ref{lem:K-multi-classification}. We can distinguish the following cases: 
\begin{enumerate}
	\item If  $(1-p)\tau_D/p<1$, we replace the right side of Eq.~(\ref{eqn:lemm1_5}) by its upper bound, i.e., 
	\begin{align}\label{eqm:lemma1_7}
	&	 \frac{ (1-p)(\vec y_0(\sigma)-\vec y_1(\sigma))}{ \frac{p}{2} + (1-p)\vec y_1(\sigma)} > 8\frac{1-p}{p}\tau_D   \nonumber\\
	\Leftrightarrow 	 &\frac{ (\vec y_0(\sigma)-\vec y_1(\sigma))}{ 4 + 8\frac{(1-p)}{p}\vec y_1(\sigma)} > \tau_D~;
	\end{align}
	\item If  $(1-p)\tau_D/p > 1$, or equivalently $p\in(0,1/2)$, we replace the right side of Eq.~(\ref{eqn:lemm1_5}) by its upper bound, i.e.,  
	\begin{align}\label{eqm:lemma1_8}
	&	 \frac{ (1-p)(\vec y_0(\sigma)-\vec y_1(\sigma))}{ \frac{p}{2} + (1-p)\vec y_1(\sigma)} > 8\frac{(1-p)^2}{p^2}\tau_D^2   \nonumber\\
	\Leftrightarrow 	 &\frac{ (\vec y_0(\sigma)-\vec y_1(\sigma))}{ 4\frac{(1-p)}{p} + 8\frac{(1-p)^2}{p^2} \vec y_1(\sigma)} > \tau_D^2~.     
	\end{align}   
\end{enumerate}
The definition $B \equiv \vec{y}_0(\sigma)/\vec{y}_1(\sigma)$ implies  
\begin{align}\label{eqm:lemma1_4}
  \vec y_0(\sigma)= B \vec y_1(\sigma)  
\Leftrightarrow   \vec y_0(\sigma)=B/(1+B).
\end{align}
Inserting this into Eqs.~\eqref{eqm:lemma1_7} and ~\eqref{eqm:lemma1_8} we have 
\begin{align}
	\frac{B-1 }{ 4(B+1) + 8\frac{(1-p)}{p}} > \tau_D
\end{align}
for the first case  $p\in(0,1/2)$, and  
\begin{align}
	\frac{B-1}{ 4\frac{(1-p)}{p}(B+1) + 8\frac{(1-p)^2}{p^2} } > \tau_D^2
\end{align}
for the second case $p\in[1/2,1)$. 
\section{Numerical simulation details} \label{app:h}
In this section, we explain how our quantum classifier is implemented and then use a generic metric to evaluate the performance of our defense protocol.
\subsection{Implementation of quantum classifier}
Our quantum classifier is composed of four main ingredients, i.e., the unitary $U_{\text{prep}}$ for state preparation, the parameterized quantum circuits $U(\bm{\theta})$ where $\bm{\theta}$ are to be optimised, the final projective measurements $|0\rangle \langle 0|$ and $|1\rangle \langle 1|$ in the $\sigma_z$ basis and the depolarisation channel $\mathcal{N}_{p}$ that is conditionally applied at testing time. Note that it doesn't matter where the depolarisation channel is placed in the circuit since results only depend on the product $\prod_{i=1}^l(1-p_i)$, where $p\equiv 1-\prod_{i=1}^l(1-p_i)$. Our circuit is shown in Figure~\ref{fig:QNN_Iris}, composed of two qubits, where each entry of the classical input vector is separately encoded into the amplitude of the quantum state in the computational basis. The state preparation unitary  $U_{\text{prep}}$, i.e., the computation of parameters $\{\vec x_i\}_{i=1}^5$, follows from \cite{plesch2011quantum}. This $U(\bm{\theta})$ is composed of 5 `layers', where each layer consist of trainable single-qubit gates and two-qubits gates as shown in the upper right panel of Figure~\ref{fig:QNN_Iris}, highlighted by the dashed box. The layers are then sequentially applied to form $U(\bm{\theta})$ \cite{benedetti2019parameterized}. The mathematical representation of $U(\bm{\theta}_{i,1}) = R_Z(\bm{\theta}_{i,1})R_Y(\bm{\theta}_{i+1,1})R_Z(\bm{\theta}_{i+2,1})$ and the total number of trainable parameters is $25$. 

\begin{figure}
	\centering
\includegraphics[width=0.48\textwidth]{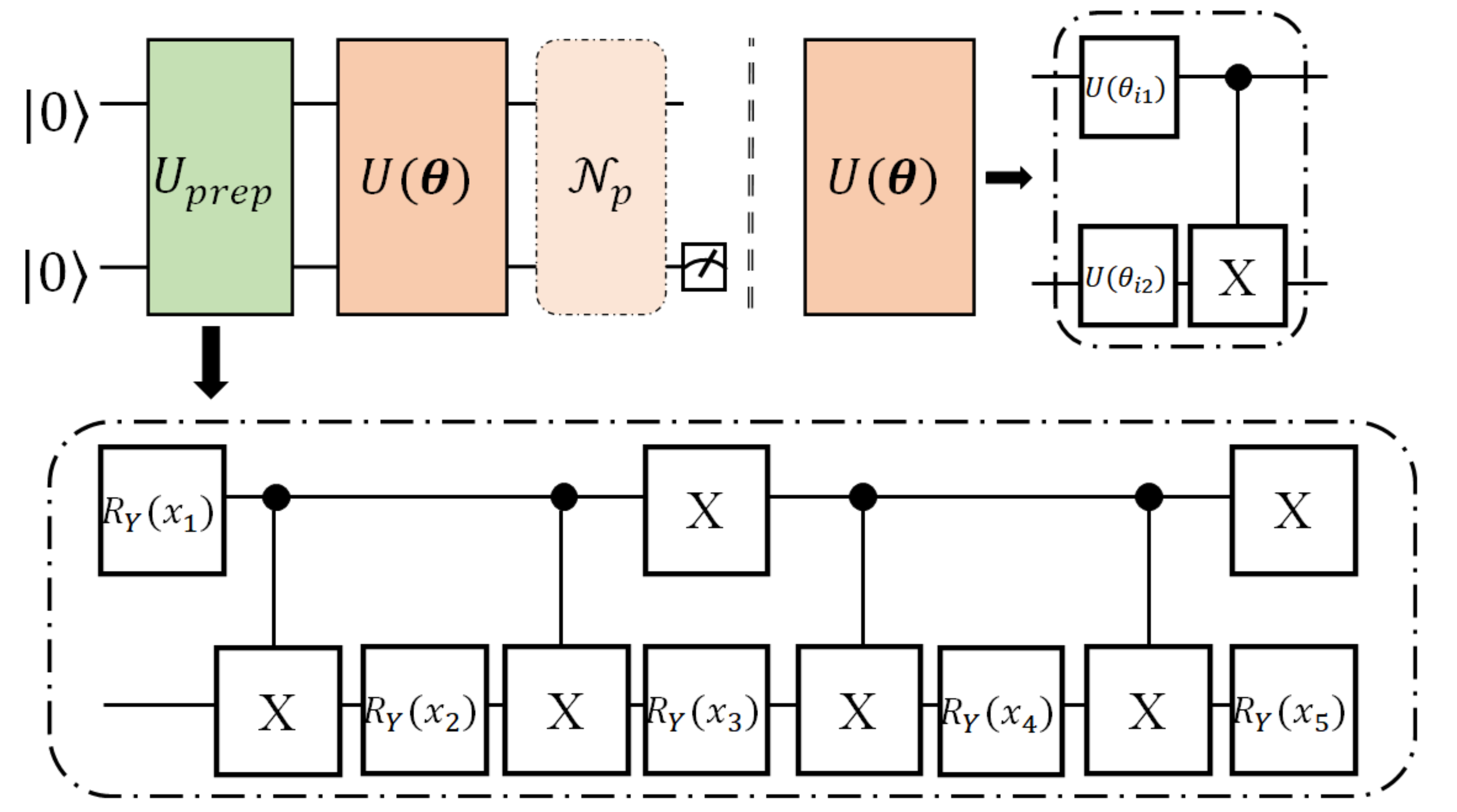}
\caption{\small{\textit{Our binary QNN classifier.} The upper left panel shows the main structure of the quantum classifier and the upper right panel illustrates one `layer' in the implementation of the trainable unitary and we employ 5 `layers' in total. The lower panel is the quantum circuit that encodes the classical input data into a quantum state.}}
\label{fig:QNN_Iris}
\end{figure}

\subsection{Evaluation}
An evaluation metric broadly used in classical adversarial learning is the conventional accuracy, which measures the prediction accuracy of the test dataset under adversarial attacks with respect to different bounded-norms \cite{lecuyer2019certified,wong2018provable}. The mathematical expression for the conventional accuracy $Acc_c$ is  
 \begin{equation}\label{eqn:Acc_c}
 	Acc_c = \frac{\sum_{i=1}^{|D_{\text{Te}}|} \bf{1}_{\bar{c}_i=c^*_i}}{|D_{\text{Te}}|}~,
 \end{equation}
where $|D_{\text{Te}}|$ is the size of the test dataset, $\bar{c}_i$ and $c^*_i$ are the predicted and real labels the $i^{\text{th}}$ test example. Here $\bf{1}_{\bar{c}_i=c^*_i}$ is the indicator function, which takes the value `$1$' when $ \bar{c}_i=c^*_i$ and is `$0$' otherwise. 
\begin{figure}[h!]
\centering 
\includegraphics[width=0.49\textwidth]{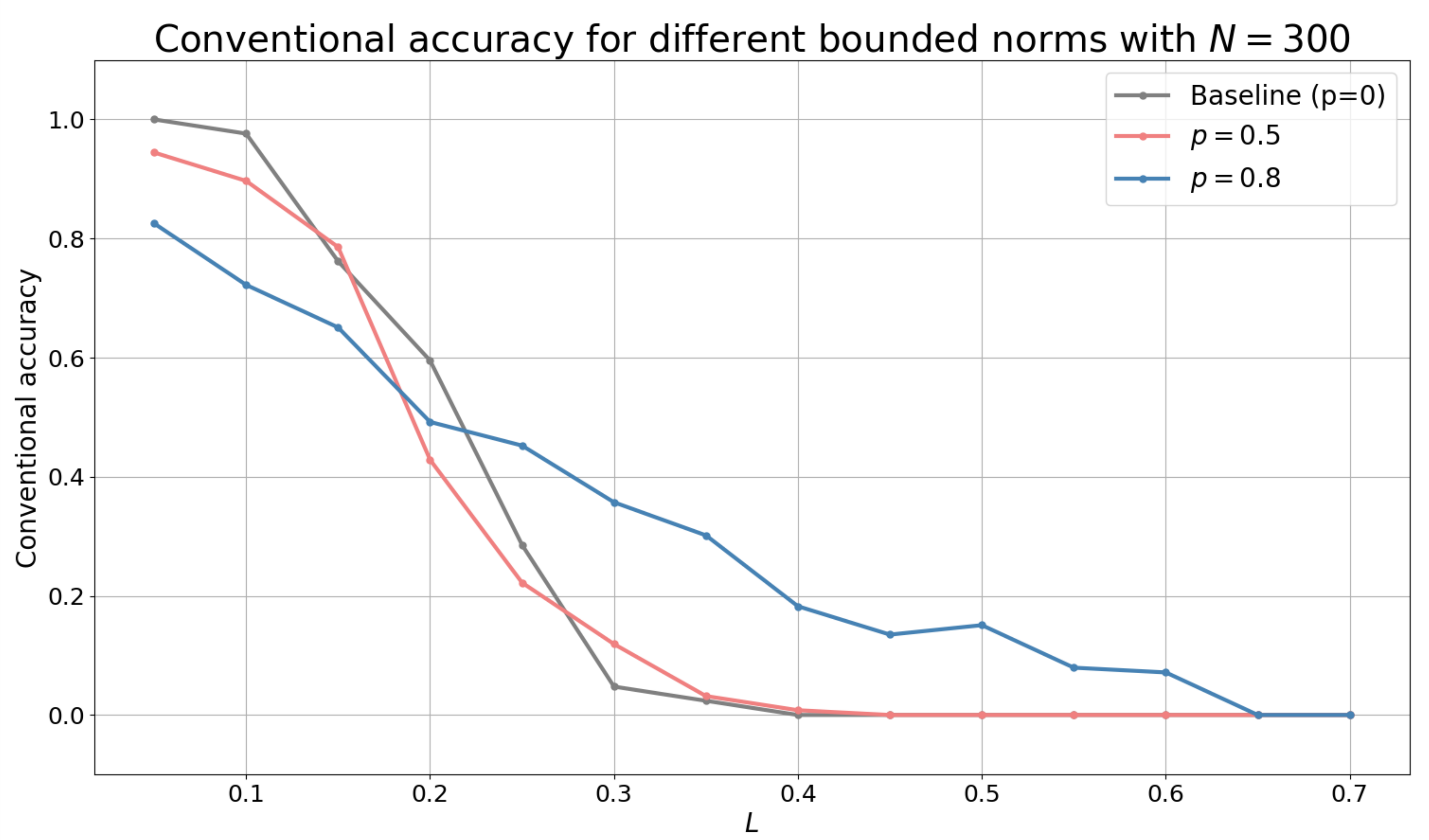} 
\caption{\small{\textit{Conventional accuracy for different depolarisation noise $p$}. We denote $L$ as the maximum $l_2$ bounded-norm used in the adversarial attack.  The conventional accuracy corresponding to $p=0.5, 0.8$ is with respect to $L$ is in red and blue respectively. The label `baseline' refers to the conventional accuracy with when $p=0$.}}
\label{fig:conv_acc} 
\end{figure}
Using the depolarisation noise $p= 0.5, 0.8$ and $\tau_D=0.015$, we explore the tradeoff between adversarial robustness and the conventional accuracy for our classifier. Let $L\in(0, 0.7]$ and $n_{samp}=300$. The number of iterations used to generate adversarial attacks is set to $50$ without early stopping.  Figure~\ref{fig:conv_acc} illustrates the simulation results under $p=0, 0.5, 0.8$. We can see how our protocol increases the robustness against $l_2$ norm attacks with increasing $p$. For instance, the conventional accuracy of our baseline ($p=0$) drops to zero when $L=0.4$, while the conventional accuracy remains non-zero for both $p=0.5$ and $p=0.8$. In addition, a larger depolarisation noise $p$ promises a better robustness against large $L$. Specifically, when $L=0.1$, the conventional accuracy when $p=0.8$ is slightly less than when $p=0.5$. However, with increased $L$, the conventional accuracy when $p=0.8$ outperforms the case when $p=0.5$. Also when $L=0.5$, both baseline and $p=0.5$ cases have the zero conventional accuracy, while the setting $p=0.8$ gives non-zero conventional accuracy.


\end{document}